\documentclass{article}

\usepackage[utf8]{inputenc}
\usepackage{authblk}
\usepackage{graphicx}
\usepackage{amsmath}
\usepackage{amssymb}
\usepackage{amsthm}
\usepackage{color}
\usepackage{url}
\usepackage{hyperref}
\usepackage{enumitem}
\usepackage{stmaryrd}
\usepackage{tikz}
\usepackage{setspace}
\usepackage{bm}
\usepackage{pdfpages}

\usepackage{caption}
\usepackage{subcaption}

\usepackage{algpseudocode}
\usepackage{algorithm}

\makeatletter
\newenvironment{subroutine}[1][htb]
{\renewcommand{\ALG@name}{Subroutine}
\begin{algorithm}[#1]%
}{\end{algorithm}}
\makeatother

\newcommand{\sub}[1]{\hyperref[alg:#1]{Subroutine~\ref*{alg:#1}}}

\newcommand{\be}{\begin{equation}}
\newcommand{\ee}{\end{equation}}
\newcommand{\ba}{\begin{array}}
\newcommand{\ea}{\end{array}}
\newcommand{\bea}{\begin{eqnarray}}
\newcommand{\eea}{\end{eqnarray}}

\DeclareMathOperator{\rank}{rank}

\newcommand{\ifrac}[2]{\left\lfloor \frac{#1}{#2} \right\rfloor} 

\newcommand{\calN}{{\cal N }}

\newcommand{\calS}{{\cal S }}
\newcommand{\calT}{{\cal T }}

\newcommand{\zeroc}{\bm{0}}

\newcommand{\decisionTree}{T}
\newcommand{\tree}{\decisionTree_\text{seen}}
\newcommand{\leaves}{\decisionTree_\text{live}}

\renewcommand{\ker}[1]{\mathsf{ker}{(#1)}}

\newcommand{\set}[1]{\{#1\}}
\newcommand{\card}[1]{\lvert#1\rvert}

\newcommand{\Prob}{{\mathbb P}}

\newcommand{\synht}{h}

\newtheorem{dfn}{Definition}

\newtheorem{lemma}{Lemma}

\newcommand{\eq}[1]{Eq.~(\ref{eq:#1})}
\renewcommand{\sec}[1]{\hyperref[sec:#1]{Section~\ref*{sec:#1}}}
\newcommand{\app}[1]{\hyperref[app:#1]{Appendix~\ref*{app:#1}}}
\newcommand{\ssec}[1]{\hyperref[ssec:#1]{Subsection~\ref*{ssec:#1}}}
\newcommand{\fig}[1]{\hyperref[fig:#1]{Figure~\ref*{fig:#1}}}
\newcommand{\tab}[1]{\hyperref[table:#1]{Table~\ref*{table:#1}}}
\newcommand{\lem}[1]{\hyperref[lem:#1]{Lemma~\ref*{lem:#1}}}
\newcommand{\propos}[1]{\hyperref[propos:#1]{Proposition~\ref*{propos:#1}}}
\newcommand{\thm}[1]{\hyperref[thm:#1]{Theorem~\ref*{thm:#1}}}
\newcommand{\alg}[1]{\hyperref[alg:#1]{Algorithm~\ref*{alg:#1}}}
\newcommand{\defn}[1]{\hyperref[defn:#1]{Definition~\ref*{defn:#1}}}

\title{Decision-tree decoders for general quantum LDPC codes}

\author[1]{Kai R. Ott}
\author[2]{Bence Hetényi}
\author[2]{Michael E. Beverland}
\affil[1]{ETH}
\affil[2]{IBM Quantum}

\begin{document}

\maketitle

\begin{abstract}
We introduce Decision Tree Decoders (DTDs), which rely only on the sparsity of the binary check matrix, making them broadly applicable for decoding any quantum low-density parity-check (qLDPC) code and fault-tolerant quantum circuits.
DTDs construct corrections incrementally by adding faults one-by-one, forming a path through a Decision Tree (DT). 
Each DTD algorithm is defined by its strategy for exploring the tree, with well-designed algorithms typically needing to explore only a small portion before finding a correction.
We propose two explicit DTD algorithms that can be applied to any qLDPC code:
(1) A provable decoder: Guaranteed to find a minimum-weight correction. 
While it can be slow in the worst case, numerical results show surprisingly fast median-case runtime, exploring only $w$ DT nodes to find a correction for weight-$w$ errors in notable qLDPC codes, such as bivariate bicycle and color codes. 
This decoder may be useful for ensemble decoding and determining provable code distances, and can be adapted to compute all minimum-weight logical operators of a code.
(2) A heuristic decoder: Achieves higher accuracy and faster performance than BP-OSD on the gross code with circuit noise in realistic parameter regimes.
\end{abstract}

\newpage
\tableofcontents

\setlength{\parskip}{\medskipamount} 

\newpage

\section{Introduction and summary of results}
\label{sec:intro}

Fault-tolerant quantum computing (FTQC) is believed to be essential for realizing large-scale quantum algorithms capable of solving classically intractable problems~\cite{beverland2022assessing,dalzell2023quantum,ho2024quantumcomputingclimateresilience,Ma_2020,Jumper_Evans_Pritzel_Green_Figurnov_Ronneberger_Tunyasuvunakool_Bates_Žídek_Potapenko_etal._2021,klusch2024quantumartificialintelligencebrief}.
FTQC relies on Quantum Error Correction (QEC) codes, which encode computational states and uses carefully designed fault-tolerant (FT) circuits to detect and correct faults during computation.
A critical challenge in this process is to determine appropriate corrections based on measurement outcomes.

This work focuses on \emph{decoders}, classical algorithms that correct faults using measurement outcomes, which must be fast and accurate for practical use.
Decoding is formulated as a linear algebra optimization problem involving two key matrices: the \emph{check matrix} $H \in \mathbb{F}_2^{M \times N}$, which maps $N$ possible faults to $M$ check outcomes, and the \emph{logical action matrix} $A \in \mathbb{F}_2^{K \times N}$, which maps faults to their effects on stored information.
Given a set of unknown faults $F \in \mathbb{F}_2^N$, the decoder uses the observed syndrome $\sigma = HF \in \mathbb{F}_2^M$ to propose a correction $\hat{F}$ such that $\sigma = H\hat{F}$.
The decoder succeeds if $\hat{F}$ satisfies both $H\hat{F} = \sigma$ and $A\hat{F} = AF$.
An important class is \emph{min-weight decoders}, which output a correction $\hat{F}$ with minimum weight, resulting in strong practical performance and provable correction guarantees. 

This framework applies to quantum codes, FT circuits, and classical codes, with instances and noise models set by $H$ and $A$. 
Unlike classical decoding, where the code (and thus $H$) can typically be selected to simplify decoding, fault-tolerant circuit decoding offers limited control over $H$, which depends on hardware and noise but is usually sparse.   
Even when $H$ has exploitable structure, prototyping benefits from testing new quantum codes and fault-tolerant strategies before developing bespoke decoders.
Thus, a compelling goal is to develop a decoder that:
\begin{itemize}[noitemsep] 
\item[(i)] \emph{works as a general qLDPC decoder (i.e., requires only that $H$ is sparse),}  
\item[(ii)] \emph{guarantees a min-weight correction, and}  
\item[(iii)] \emph{runs in poly($w$) worst-case time for weight-$w$ faults, enabling fast decoding in practical settings.}   
\end{itemize}

We expect there is no decoder which meets all three criteria, although to date its existence has not been ruled out\footnote{NP-hardness is known for min-weight decoding of general stabilizer~\cite{hsieh2011} but is not known for the qLDPC sub-family.}.
MaxSAT reductions enable general qLDPC decoders with guaranteed min-weight corrections but remain much slower than other decoders in practical settings, even with optimized SAT solvers~\cite{Berent_2024, noormandipour2024maxsatdecodersarbitrarycss}.
Min-weight perfect matching~\cite{dennis2002topological} and union-find~\cite{delfosse2021almost} decoders satisfy (ii) and (iii), but not (i) because they are limited to check matrices with a maximum column weight of two—applicable to surface codes (and to color codes via reduction~\cite{delfosse2014decoding}) but excluding many other relevant cases.
Notably, union-find has been extended to some other qLDPC codes, but its correction guarantees weaken significantly~\cite{delfosse2022toward}, failing criterion (ii).
Formally efficient, provable decoders for qLDPC codes with finite rate~\cite{panteleev2021quantum,krishna2024viderman} and finite relative distance~\cite{dinur2023good,lin2022good} represent significant theoretical progress.  
However, these decoders require check matrices with strong expansion properties~\cite{leverrier2015quantum,krishna2024viderman}, limiting their applicability regarding criterion (i).  
Moreover, their practical use is hindered by high runtime prefactors and large code sizes~\cite{breuckmann2021ldpc,stambler2023addressing}.

The absence of practically efficient, provable decoders for general qLDPC codes has led to heuristic decoders, such as Belief Propagation with Ordered Statistics Decoding (BP-OSD)~\cite{fossorier2002soft,panteleev2021degenerate,roffe2020} and a range of recent proposals~\cite{higgott2023improveddecodingcircuitnoise,javed2024low,gong2024,iolius2024closed,wolanski2024ambiguity,demarti2024almost,kung2024efficient}.
Heuristic decoders are currently the only viable option for many practical applications, but replacing them with provable decoders that achieve competitive runtimes (satisfying (i), (ii), and (iii)) would be highly desirable.
This work makes progress toward that goal, with key contributions summarized below.

\begin{figure}[ht]
  \centering
    \includegraphics[height=5cm]{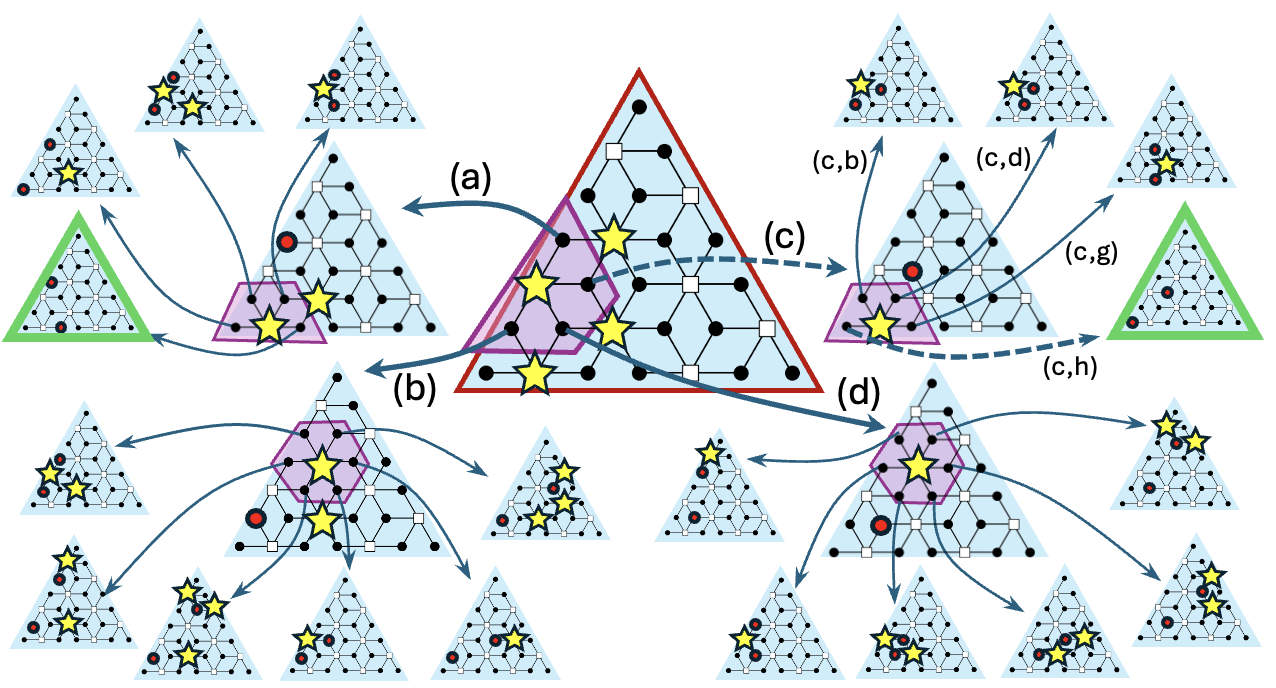}
    \hspace{0.3cm}\includegraphics[height=4.4cm]{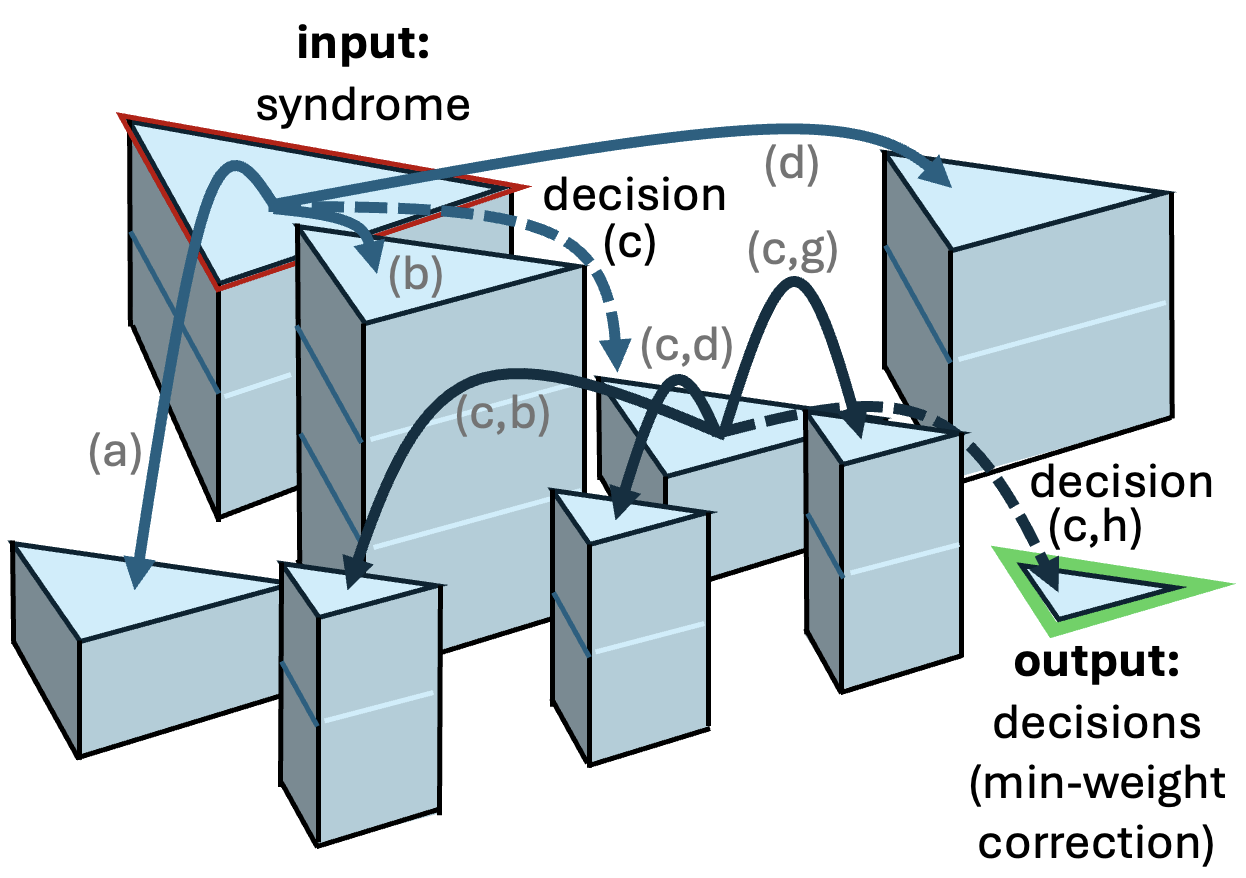}
     \caption{
    \textbf{(Left)} An illustration of the first three levels of the decision tree (DT) for a distance-5 color code.
    Each DT node (blue triangle) represents a partial correction, with the root DT node (red outline) at the center.
    DT nodes are labeled with the code’s Tanner graph, where black circles are data vertices, white squares are check vertices, red circles indicate partial corrections, and yellow stars mark syndrome vertices.
    A syndrome vertex is selected from each DT node, and its neighboring fault vertices are added to form child nodes.
    DT nodes with trivial syndromes (green outline) represent valid corrections.
    \textbf{(Right)} Starting at the root node, DT decoders iteratively explore by assigning a cost to each child of the lowest-cost node in the tree. 
    For example, using the syndrome height $h(\sigma)$, the minimum weight of any error with syndrome $\sigma$, as the cost function yields a minimum-weight correction in $w$ steps for weight-$w$ errors.
     }
  \label{fig:intro-figure}
\end{figure}

\newpage
\textbf{Main results: }
We introduce the general qLDPC Decision Tree Decoder (DTD) family, including:
\vspace{-\parskip}
\begin{enumerate}[noitemsep]
    \item[(1)] A provable decoder that satisfies (i) and (ii). 
    Although it generally falls short of (iii), numerical evidence suggests fast median-case runtime for some qLDPC code families.
    \item[(2)] A heuristic decoder with better accuracy and faster empirical runtime than BP-OSD.
\end{enumerate}

The provable decoder is likely our most significant contribution, as few general qLDPC decoders offer correction guarantees~\cite{Berent_2024, noormandipour2024maxsatdecodersarbitrarycss}.

\textbf{Decision Tree Decoders: }
The core idea of the DTD family is to construct a correction incrementally, adding faults one at a time, akin to the flip algorithm of Sipser and Spielman~\cite{sipser1996expander}. 
This decision sequence traces a path through a \emph{decision tree} (see \fig{intro-figure}), retaining all possible choices, in contrast with list decoding algorithms, where choices are dropped~\cite{dumer2006soft,tal2015list}.
Structural properties of the problem help reduce the search space; for instance, each unsatisfied check vertex neighbors at least one actual fault: it is sufficient to restrict to faults connected to a single check vertex at each search step. 
Nonetheless, the decision tree grows exponentially, and effective algorithms must explore and construct only a small part before identifying a low-weight correction that cancels the syndrome. 
Different DTD algorithms use varying strategies to navigate the decision tree.
DTDs have a loose sparsity requirement for the check matrix $H$ to keep the number of children per node manageable.  

A key observation (see \fig{intro-figure}) is that a DTD decoder satisfying criteria (i), (ii), and (iii) could be constructed with an oracle for the \emph{syndrome height} $h(\sigma)$, defined as the minimum weight of a correction for a given $\sigma$.  
In each round, the decoder selects a fault that reduces the syndrome height by one, which yields a min-weight correction after precisely $h(\sigma)$ iterations.  
While creating a practical algorithm this way seems unlikely, as no efficient method for computing $h(\sigma)$ is known for general qLDPC codes, this idea motivates both our provable and heuristic DTD algorithms.

\textbf{Provable decoder: }    
The Height-bound DTD algorithm performs an assisted breadth-first search, using easily computable lower bounds on $h(\sigma)$ to prune large parts of the decision tree.
We use `neighborhood bounds' applicable to any qLDPC code, which rely on the requirement that each unsatisfied check vertex has at least one neighboring fault.
Height-bound DTD also uses BP for tie-breaking, improving search efficiency while ensuring min-weight corrections.
Empirically, it explores an optimal number of decision tree nodes in the median case across various QEC codes (see \fig{minweight-main-results}), including color codes~\cite{bombin2006topological,Bombin2007} and bivariate bicycle codes~\cite{bravyi2024high}.
This is encouraging, as the decoder performs well despite weaker syndrome-neighborhood bounds for topological codes like color codes, where syndromes appear only at string endpoints.
Even better performance is expected for expander codes, where large errors produce large syndromes.
An intuitive explanation for height-bound DTD's strong median-case performance, even for topological codes, is that while large errors with small syndromes (and thus weak syndrome-neighborhood bounds) occur, they are quite rare.

\begin{figure}[ht]
  \centering
    \includegraphics[width=0.75\textwidth]{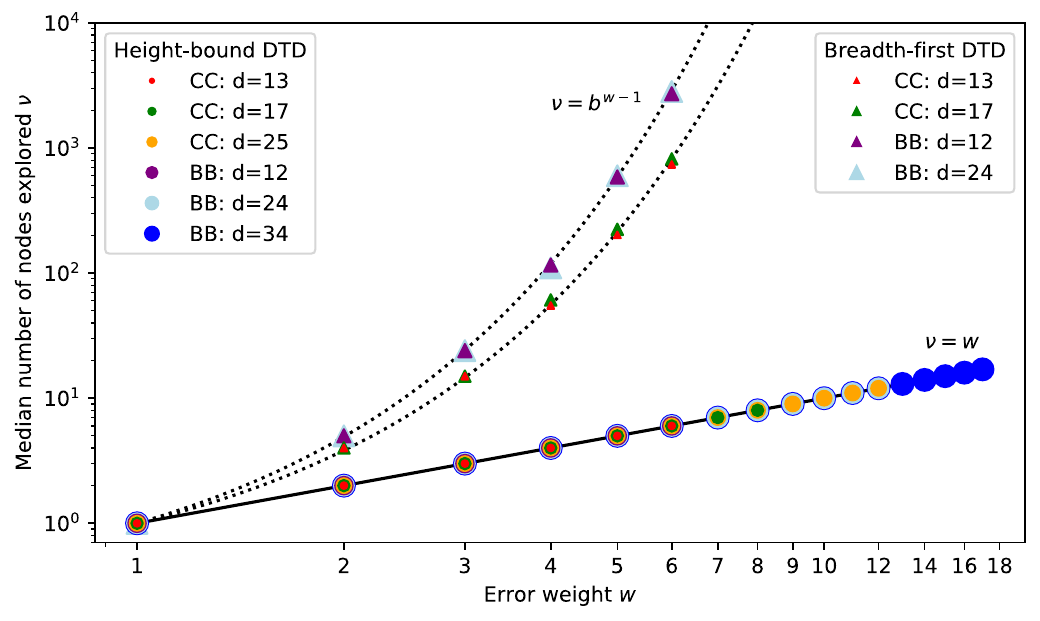}
     \caption{
     \textbf{Provable decoder.}
     For several distance-$d$ color codes (CC) and bivariate bicycle codes (BB), for each $w \leq \frac{d-1}{2}$ we randomly sample all weight-$w$ $X$-type faults.  
    We report the median number of decision tree nodes explored by two decoding algorithms: height-bound DTD and naive breadth-first DTD.  
    As expected, breadth-first DTD explores exponentially many nodes in $w$ (with dashed, fit parameter $b_\text{CC}=3.825$ and $b_\text{BB}=4.902$), while height-bound DTD, perhaps surprisingly, explores only $w$ nodes. 
     }
  \label{fig:minweight-main-results}
\end{figure}

\textbf{Heuristic decoder: }  
Unlike the provable decoder height-bound DTD, which performs a rigorously guided breadth-first search of the DT, our heuristic algorithm BP-DTD uses belief propagation to estimate costs for DT nodes and explores in a depth-first manner.
Another difference is that BP-DTD accommodates non-uniform fault probabilities, making it better-suited for circuit noise than height-bound DTD.
This approach offers faster worst-case performance than height-bound DTD but does not guarantee min-weight corrections.
Our fast, non-provable BP-DTD algorithm aligns with several recently proposed heuristic decoders~\cite{higgott2023improveddecodingcircuitnoise,javed2024low,gong2024,iolius2024closed,wolanski2024ambiguity,demarti2024almost,kung2024efficient,lobl2024breadth,hillmann2024localized,scruby2023local}, and is conceptually most similar to the guided-decimation decoder~\cite{gong2024} and the closed-branch decoder~\cite{iolius2024closed}.   
Although detailed optimization and comparisons with other heuristic decoders are beyond the scope of this work, encouraging initial results in \fig{ler-vs-tcutoff-intro} show significant improvements over BP-OSD for time-sensitive decoding.

\begin{figure}[ht]
  \centering
    \includegraphics[width=0.8\textwidth]{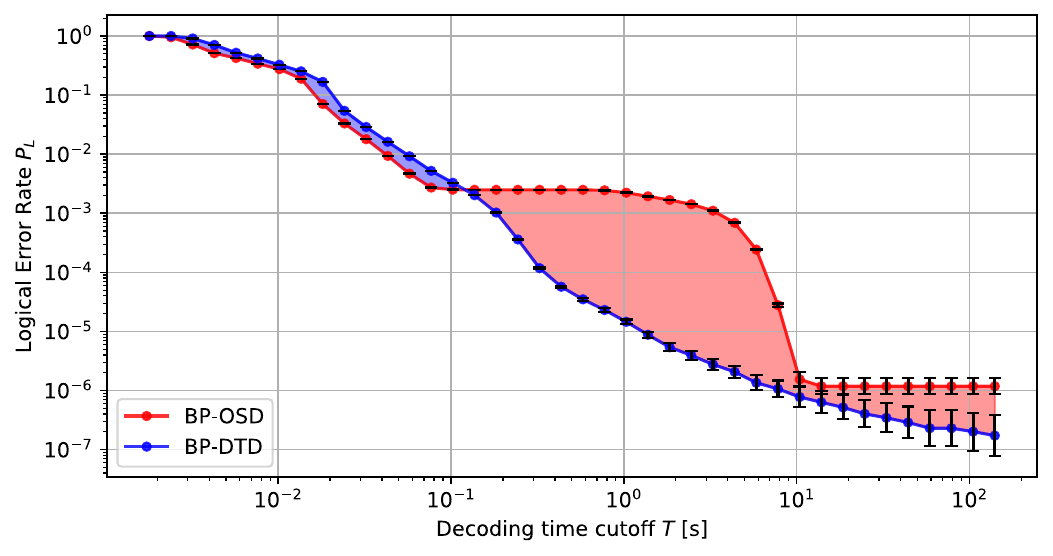}
     \caption{
     \textbf{Heuristic decoder.}
     Cutoff-time performance curves for BP-OSD and BP-DTD for the gross code under circuit noise of strength $p = 10^{-3}$. 
     Failure probability arises from two sources: exceeding the cutoff time $T$ or decoding within $T$ but producing an incorrect correction.
     At early and late times, both decoders perform similarly—either relying on a BP pre-decoder or having enough time to terminate, with BP-DTD showing slightly lower logical error rates in the late regime.  
     In the intermediate regime, BP-DTD terminates more often than BP-OSD, leading to significantly lower logical error rates.
     }
  \label{fig:ler-vs-tcutoff-intro}
\end{figure}

\textbf{Applications and open questions: }
The height-bound decoder provides provable minimum-weight corrections and often runs very quickly, though it can be slow in some cases. 
This makes it appropriate for decoder ensembles~\cite{sheth2020neural,bausch2023learning,shutty2024efficient}, where a fallback decoder could handle cases where it fails to terminate on time. 
Beyond decoding, the decision tree and height bound are valuable tools. 
In \app{find-min-weight-logicals}, we describe an algorithm that leverages these tools to compute all minimum-weight logical operators of a code.
In \app{finding-code-distances} we review how to use height-bound DTD (or any min-weight decoder) as a subroutine in distance-finding algorithms by finding minimum-weight corrections for an augmented check matrix as has been explored with SAT-solvers~\cite{shutty2022decoding}.  
Additionally, BP-DTD appears well-suited for gateware implementations, like FPGAs, due to its parallelizability and the theoretical $O(1)$ runtime per BP iteration on specialized hardware.  
Open questions remain (see \sec{conclusion}): how can these DTD algorithms be improved, what other provable or practical DTD algorithms can be developed, and what are their most compelling applications?

The remainder of this work is organized as follows:  
In \sec{decoding-scenarios}, we review how the decoding matrices $H$ and $A$ arise in various contexts, including classical codes, stabilizer codes, and quantum circuits.  
\sec{decoding} provides key definitions, and \sec{main-algo} introduces the general decision-tree decoder family.  
In \sec{provable-cost} and \sec{fast-cost} we present and analyze our provable and heuristic decoding algorithms height-bound DTD and BP-DTD. 
We conclude in \sec{conclusion}, with a brief discussion of potential future work.  
Additional details and results are provided in the appendices.

\clearpage

\section{A guide to decoding matrices in different settings}
\label{sec:decoding-scenarios}

This section provides a unified language to describe decoding in terms of decoding matrices $H \in \mathbb{F}_2^{M \times N}$ and $A \in \mathbb{F}_2^{K \times N}$ across various contexts, as summarized in \tab{decoding-matrices-summary}. 
Our decision tree decoders are applicable to all of these scenarios, provided that the check matrix $H$ is sparse. 
While this section aims to offer a helpful clarification and review of decoding, it is not essential for understanding the remainder of the paper; key definitions and notation are established in \sec{decoding}.


\begin{table}[h!]
\centering
\begin{tabular}{ |p{3.75cm}||p{2.4cm}|p{2.45cm}|p{5.65cm}|  }
 \hline
\textbf{Decoding problem} & \textbf{Check matrix}  & \textbf{Logical action}  & \textbf{Comments} \\
  & ($M \!\times\! N$) & \textbf{matrix} ($K \!\times\! N$)\!\!  &  \\
 \hline
 \hline
 Classical code $[n,k,d]$ & $ m \times n $    & $ n \times n $ &  $A$ is identity matrix (usually omitted) \\
 \hline
 Stabilizer code $[[n,k,d]]$ & $ m \times 2n $    & $ 2k \times 2n $&  $m \geq n-k$ ($>\!$ if over-complete checks$^*\!$)\!\\
 \hline
   &&& Constraint $(m_X + m_Z) \geq n - k$. \\
  CSS-type & $H_X$:~$m_X \times n $    & $A_X$:~$ k \times n $ & Special case of stabilizer code with:  \\
  stabilizer code $[[n,k,d]]$ & $H_Z$:~$m_Z \times n $    & $A_Z$:~$ k \times n $ & $H = \begin{pmatrix}
H_X & 0 \\
0 & H_Z 
\end{pmatrix}$, $A = \begin{pmatrix}
A_X & 0 \\
0 & A_Z 
\end{pmatrix}$ \\
 \hline
  Logical memory circuit & $m R \times N $    &  $2k \times N $ & $R$-rounds of stabilizer extraction.\\ 
  for stabilizer code $\![[n\!,k\!,d]]\!\!\!\!$ & & &   $N$ depends on circuit and noise. \\ 
 \hline
 Logical operation  &&& Measures $l$ commuting logical Paulis, \\
 circuit for stabilizer  &  $M \times N$  & $(k+k') \times N$  & and applies $k-l$ to $k'$  logical isometry. \\
 codes $\![[n\!,k\!,d]]\!\!\rightarrow \!\![[n'\!\!,k'\!\!,d']]\!\!\!$ & & & $M,N$ depend on circuit and noise. \\
 \hline
\end{tabular}
\caption{
Summary of check and logical action matrices $H \in \mathbb{F}_2^{M \times N}$ and $A \in \mathbb{F}_2^{K \times N}$ across settings.
$^*$The constraint $m \geq n-k$ applies everywhere $m$ appears in the table. 
Classical, stabilizer, and CSS codes assume perfect check measurements.
Logical memory and logical operation circuits account for general circuit noise, including noisy measurements, but assume a noise-free final round of  stabilizer measurements (this assumption can be relaxed). 
For CSS codes, we use the convention $H_X = S_Z$, with $H_X$ detecting $X$ errors and $S_Z$ the $Z$-type stabilizer matrix.
CSS-type restrictions on circuits, not shown here, similarly block-diagonalize the circuit decoding matrices, as they do for CSS codes.
}
\label{table:decoding-matrices-summary}
\end{table}

Given unknown faults $F \in \mathbb{F}_2^N$, the decoder is provided the syndrome $\sigma = HF \in \mathbb{F}_2^M$ and proposes a correction $\hat{F}$, succeeding if $H\hat{F} = \sigma$ and $A\hat{F} = AF$.
Before reviewing decoding in each of the settings in \tab{decoding-matrices-summary}, we comment on some general aspects:

\begin{itemize}

\item It is natural to define two additional matrices—the stabilizer matrix $S$ and the logical operator matrix $L$—given $H$ and $A$. 
While well-known in the context of stabilizer codes\footnote{
For a stabilizer code, $S \in \mathbb{F}_2^{n-k \times 2n}$ and $L \in \mathbb{F}_2^{2k \times 2n}$ correspond to the generators of the stabilizer group $\mathcal{S}$ and the logical operators $\mathcal{L}$, respectively. 
In this case, $H$ and $A$ align with $S$ and $L$ after swapping the first and second sets of $n$ columns. 
This direct correspondence between $H$ and $S$ and between $A$ and $L$ does not hold for circuit decoding.
}, $S$ and $L$ are also meaningful in more general decoding settings.
The stabilizer matrix $S$ consists of rows that form a complete linear basis of vectors $s$ satisfying $Hs = 0$ and $As = 0$. 
This captures the equivalence imposed by $A$: two vectors $f$ and $f'$ are equivalent ($f \sim f'$) if and only if $Hf' = Hf$ and $Af' = f'$, which is equivalently expressed as $f' = f + s$ for $s \in \text{rowspace}(S)$.
The logical operator matrix $L$ consists of rows that form a complete, linearly independent basis of inequivalent vectors $l$ satisfying $Hl = 0$ and $Al \neq 0$ (any such vector $l$ is a linear combination of rows of $L$ and $H$).

\item The decoding problem is invariant under linearly independent row recombinations of each of $H$ and $A$. 
In practice, certain choices can aid decoders (e.g. regular structure, redundancy and low row/column weights).

\item In this section, we do not restrict the probability distribution for $F$: we allow each of the $2^N$ bitstrings to occur with a given probability, capturing arbitrary correlations between the bits of $F$. 
Outside this section (see \sec{decoding}), for simplicity we restrict to uncorrelated distributions, where each bit $F_j$ is independently drawn with probability $p_j$. 
We can however introduce additional bits to $F$ to reproduce the effect of correlations between the original bits. 
For example, in a stabilizer code with noise where $X$, $Y$, and $Z$ errors occur independently with probability $p$, one could use a length-$N=2n$ fault vector (encoding $X$ and $Z$ errors with correlations) or a length-$N=3n$ fault vector (encoding $X$, $Y$, and $Z$ errors independently) to represent the same error model.

\item Throughout this work, we use the term `quantum-LDPC' to refer to any infinite family of decoding problems with (a fixed choice of) check matrices which have bounded row and column weights.

\end{itemize}

\subsection{Classical codes}
\label{sec:classical-codes}

For a classical linear code, the check matrix $H \in \mathbb{F}_2^{M \times N}$ is the main object defining the code.
Codewords are bitstrings $F \in \mathbb{F}_2^N$ such that $HF=0$.
Given a codeword $F$, an additional error $E$ results in $(F+E)$, which has syndrome $\sigma = H(F+E)=HE$.
The logical action matrix $A$ is the identity matrix and does not play a role (therefore $A$ is typically not discussed in classical coding theory). 
This means that in the typical scenario of classical codes, all codewords are considered in-equivalent, and as such the distance $d$ of the code is the Hamming weight of the smallest non-trivial vector in the kernel of $H$.
Therefore the classical code properties $[n,k,d]$ correspond to the $n=N$ bits of the message, which encodes $k= \dim \ker{H} = N-\text{rank}(H)$ logical bits.
Since $M \geq \text{rank}(H)$, we have the constraint that $M \geq n-k$.

\subsection{Stabilizer codes}
\label{sec:stab-codes}

Here we specify how the decoding matrices arise for stabilizer codes.
Let $n$ be the number of qubits that the code is defined on.
A stabilizer code is defined by its \emph{stabilizer group} $\mathcal{S}$, which is an abelian subgroup of the $n$-qubit Pauli group which does not include $-I$.
The code space is the simultaneous +1 eigenspace of every element of $\mathcal{S}$.
We typically fix a set of $m$ generators of the stabilizer group, which can be over-complete, and use this to define the code.
In \fig{example-codes}(a) and (b) we define a family of color codes~\cite{bombin2006topological,Bombin2007} and the gross code~\cite{bravyi2024high} (an example of a bi-variate bicycle code \cite{mackay2004sparse,Kovalev2013a,panteleev2021degenerate,panteleev2021quantum}) which we use as examples.

\begin{figure}[ht]
  \centering
    (a)\includegraphics[height=3.5cm]{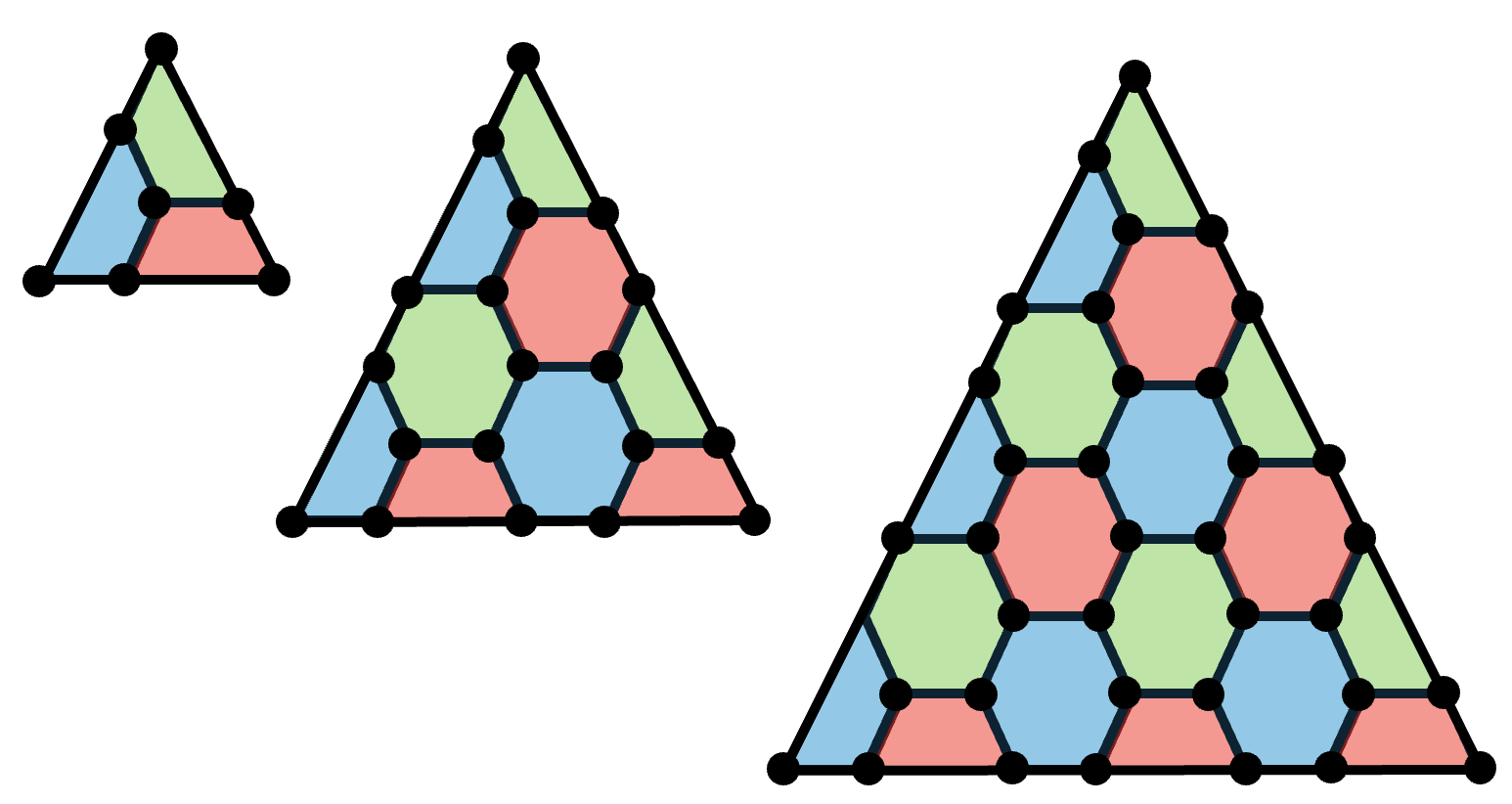}
    (b)\includegraphics[height=4cm]{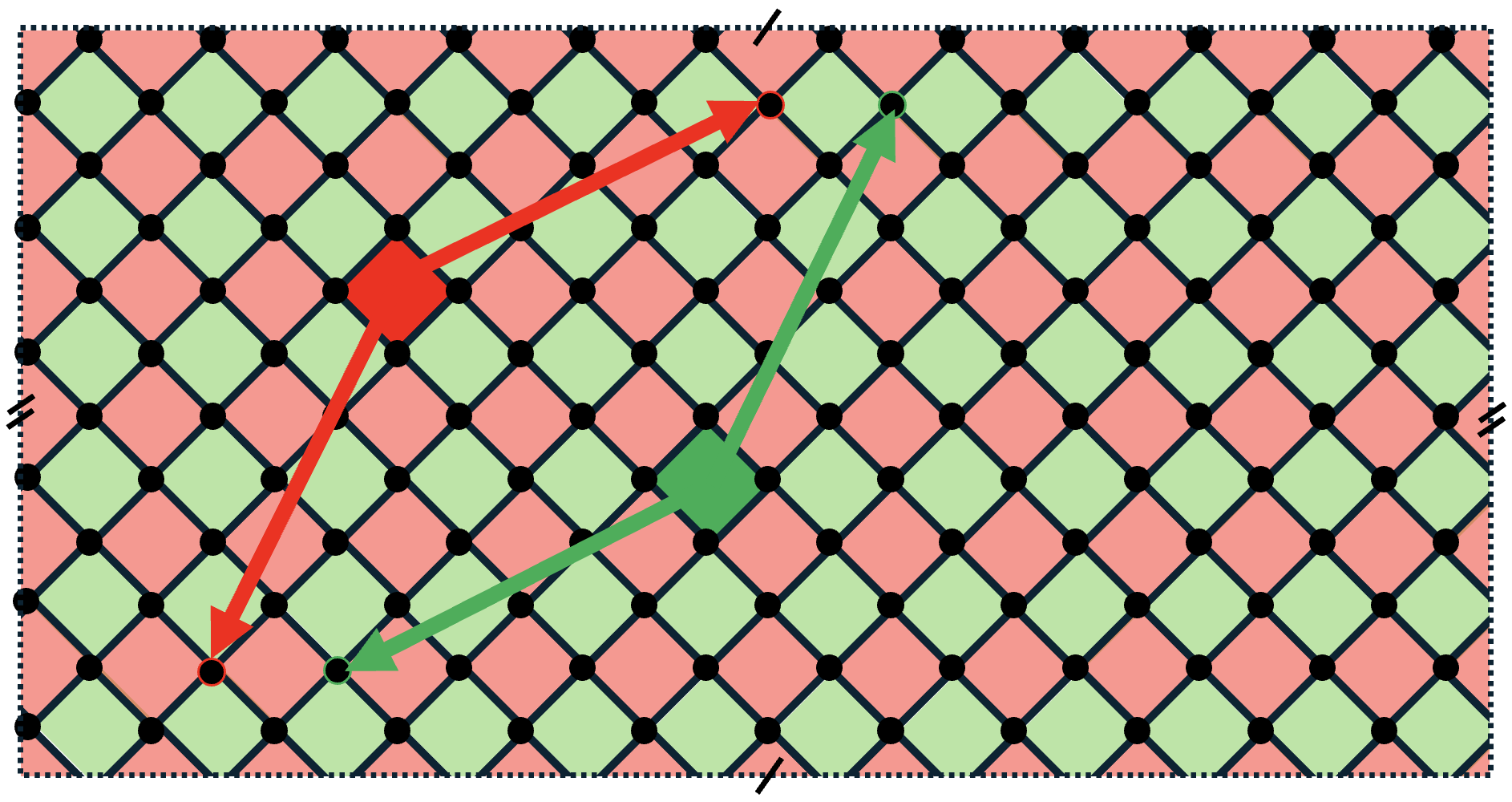}
     \caption{
     (a) Three instances of the color code family with distances 3, 5 and 7 (larger distance codes are defined by extending the pattern).
     Qubits are at vertices, with an $X$- and a $Z$-type stabilizer generator supported on the qubits of each face. 
     (b) The gross code with qubits at vertices on this tiling of the torus.
     There is a weight-6 $X$-type ($Z$-type) stabilizer generator supported on the four qubits on each red (green) square face and on two additional qubits located at fixed translated vectors relative to the face, indicated by red (green) arrows for one highlighted face.
     }
  \label{fig:example-codes}
\end{figure}

In what follows we define a number of important properties of a stabilizer code $\mathcal{S}$.
To distinguish some of these from other definitions for the more general space-time codes in the next section, we will add $\mathcal{S}$ as a superscript in places.

{\bf Checks:}
A check is a bit that captures the measurement outcome of a generator, taking a value $0$ for a code state in the absence of any error.
When a check is $0$ (corresponding to a +1 measurement outcome) we say it is `satisfied', and when it is $1$ (corresponding to a -1 measurement outcome) we say it is `unsatisfied'.

{\bf Pauli errors: } 
We can represent any $n$-qubit Pauli error $P_1 \otimes P_2 \otimes \dots P_n$ as a length $2n$ bit string $E = b_1 b_2, \dots b_{2n}$, where for $j \in 1,\dots n$ the pair of bits $(b_j,b_{n+j})$ are: $(0,0)$ if $P_j=I$, 
$(1,0)$ if $P_j=X$, 
$(0,1)$ if $P_j=Z$ and
$(1,1)$ if $P_j=Y$.

{\bf Check matrix:} 
The check matrix $H^\mathcal{S} \in \mathbb{F}_2^{m \times 2n}$ of the code is a binary matrix with $2n$ columns and $m$ rows.
(Note that $H^\mathcal{S}$ is the check matrix $H$ that is used for decoding the stabilizer code. 
We use the superscript $\mathcal{S}$ here on $H$ and $A$ to make it easier to refer back to these objects later.)
The entry $H^\mathcal{S}_{ij}$ is $0$ if the $i$th generator commutes with $X$ on the $j$th qubit, and $1$ otherwise, while the entry $H^\mathcal{S}_{i,n+j}$ is $0$ if the $i$th generator commutes with $Z$ on the $j$th qubit, and $1$ otherwise. 

The \emph{Tanner graph} $\mathcal{T}$ of the code is the graph which has the check matrix $H^\mathcal{S}$ as its adjacency matrix, and it is conventional to use circular nodes for errors and square nodes for stabilizer checks.
It is common to draw the $X$- and $Z$-type subgraphs of the Tanner graph, which we write as $\mathcal{T}_X$ and $\mathcal{T}_Z$ respectively, which correspond to the first and second sets of $n$ columns of $H^\mathcal{S}$ respectively. 
(A stabilizer generator which is neither purely $X$ nor purely $Z$ type will appear in both subgraphs).

{\bf CSS-type stabilizer codes:} 
If the code is a CSS~\cite{gottesman1996class,Calderbank1997} code, a set of generators exists which are each either purely $X$-type or purely $Z$-type, such that $H^\mathcal{S}$ is block diagonal, in which case $\mathcal{T}_X$ and $\mathcal{T}_Z$ are disjoint subgraphs of $\mathcal{T}$. 
Our convention for CSS codes is to list the $m_X$ $Z$-type stabilizer generators first, which form the top-left submatrix $H_X$ of $H^\mathcal{S}$, and to then list the $m_Z$ $X$-type stabilizer generators, which form the bottom-right submatrix $H_Z$ of $H^\mathcal{S}$.
(Note that this convention for CSS code matrices is the opposite of what is most common in the literature, but avoids the definition of the symplectic product and also is more parallel to our other definitions of check matrices which are defined by considering rows in terms of detections of faults.)
For example, the decoding graphs depicted in \fig{example-decoding-graphs} for $X$-type noise with perfect measurement correspond to the $X$-type Tanner graphs for the color and gross codes.
(Note: in some literature the Tanner graph is drawn more compactly for CSS codes by combining the node for each error pair $X_i$ and $Z_i$ such that there is a single node corresponding to each data qubit $i$, which is unambiguous since each check node corresponds to a pure $X$- or $Z$-type stabilizer generator.)

{\bf Syndrome: }  
Let $E$ be the length-$2n$ bit string that represents an $n$-qubit Pauli error. 
The \emph{syndrome} $\sigma^\mathcal{S}$ of $E$ is $\sigma^\mathcal{S} = H^\mathcal{S} E$ (with arithmetic modulo two) is the list of $m$ measurement outcomes that would be obtained if the stabilizer generators of the code were measured on a code state with $E$ applied to it. 
In other words, the syndrome is a vector of check outcomes ordered according to the rows of $H^\mathcal{S}$.

{\bf Logical action matrix: }  
Any Pauli operator which commutes with all stabilizer generators but is not in either $\mathcal{S}$ or $-\mathcal{S}$ is a non-trivial logical operator (stabilizers are sometimes referred to as trivial logical operators).
We define the distance $d$ as the smallest weight of any non-trivial logical operator.
Two logical operators $B$ and $C$ are equivalent if there is a stabilizer $S \in \mathcal{S}$ such that $B \propto S C$.
The number of independent stabilizer generators is $2k$, where $k$ is the number of logical qubits encoded by the code.
We form the logical action matrix $A^\mathcal{S} \in \mathbb{F}_2^{2k \times 2n}$, which is a binary matrix with $2n$ columns and $2k$ rows, by taking each row to be the bit string representing an independent logical operator.
It is common  use the notation $[[n,k,d]]$ to specify these important properties of a stabilizer code.

\subsection{Logical memory circuits for stabilizer codes}
\label{sec:space-time-codes}

Here we assume a circuit which measures a set of stabilizer generators of an $[[n,k,d]]$ stabilizer code with stabilizer group $\mathcal{S}$, repeated over $R$ rounds.
We will assume that the last round is fault-free (this assumption can be relaxed to form a windowed-decoding strategy, which we will not discuss in this work for simplicity).
The presentation intentionally mirrors the previous subsection on stabilizer codes to highlight analogies.

{\bf Circuit checks: }
The circuit outputs a length-$M$ bit string of \emph{circuit checks}, each of which is a bit formed from the parity of the measurement outcomes of a stabilizer generator on two consecutive rounds, which is 0 in the absence of faults.
When a circuit check is zero we say it is `satisfied', and when it is one we say it is `unsatisfied'.
In some literature~\cite{Gidney2021stim}, circuit checks are known as `detectors'.
Given $R$ rounds, and $m$ stabilizer generators of the code, there will be $M = R \cdot m$ circuit checks.

{\bf Circuit faults: }
These can be thought of as a generalization of Pauli errors for a stabilizer code.
We assume that an explicit set of $N$ distinct faults can occur.
Each of the $N$ faults has the effect of flipping a subset of detector outcomes, and leaves a residual Pauli error at the end of the circuit. 
We use the fault bit string $F \in \mathbb{F}_2^N$ to denote a set of faults (indicated by 1s in the bit string) occurred.

{\bf Circuit check matrix: }
The circuit check matrix $H \in \mathbb{F}_2^{M \times N}$ is a binary matrix with $N$ columns and $M$ rows.
For $i \in 1,\dots M$, the entry $H_{ij}$ is $1$ if the $i$th detector outcome is flipped by the $j$th fault, and $0$ otherwise.

{\bf Circuit syndrome: } 
The syndrome of a set of faults $F$ is the bit string $\sigma$ of detector values that result from $F$, which is obtained from the circuit check matrix as $\sigma = H F$ (with arithmetic modulo two).

{\bf Circuit logical action matrix: }
The circuit logical action matrix $A \in \mathbb{F}_2^{2k \times N}$ encodes the effect of the residual Pauli at the end of the circuit caused by each fault.
(Note that the logical action matrix needs to be modified if the space time code implements a non-trivial logical operation such as measuring a logical operator rather than just syndrome extraction~\cite{beverland2024fault}.)
More specifically, suppose that the $j$th fault results in a residual error $E_j$ (in binary symplectic form) at the end of the circuit.
Then the $j$th column $A_{*j}$ of $A$ is given by $A_{*j} = A^\mathcal{S} E_j$, where $A^\mathcal{S}$ is the logical action matrix of the code.
The \emph{circuit distance} $d$ is the minimum weight of any circuit logical.

We can also define a \emph{circuit stabilizer group} as follows.
Given the full set of $N$ faults that can occur, a set of faults $S \in \mathbb{F}_2^N$ is a circuit stabilizer if and only if $HS=0$ and $AS=0$. 
The set of all such faults forms a group.

\subsection{Logical operation circuits for stabilizer codes}
\label{sec:stabilizer-channels}

Here we review decoding for a more general setting where quantum circuits are used to build logical operations on stabilizer codes.
We do not test our decoders on logical operations in this paper, but we hope that this description may be useful for reference.

The general description of logical operation circuits we review here has been considered in many different works and is referred to by many names, including the \emph{detector noise model}~\cite{Gidney2021stim}, \emph{spacetime codes}~\cite{delfosse2023spacetime,bacon2017sparse,gottesman2022opportunities}, \emph{logical blocks}~\cite{Bombin2023}, and \emph{stabilizer channels}~\cite{beverland2024fault}.
All of these works include descriptions of logical operation circuits which are somewhat equivalent to what we review here, but our notation and presentation most closely aligns with the material in Refs.~\cite{beverland2024fault,kliuchnikov2023stabilizer}.

\textbf{Stabilizer channels: }
A \emph{stabilizer channel} is a circuit composed of \emph{stabilizer operations}—Pauli basis state preparations, measurements, Clifford unitaries, and conditional Pauli measurements based on parities of earlier measurement outcomes. 
Such a circuit transforms an initial $[[n,k,d]]$ stabilizer code into a (possibly different) output $[[n',k',d']]$ stabilizer code. 
In this context, a decoder's task is to detect and correct faults in the circuit, ensuring robust logical action.
We will assume in the rest of this discussion that the circuit is a stabilizer channel.

\textbf{Logical action: }
The logical action of any (stabilizer channel) circuit has a simple form~\cite{beverland2024fault} as shown in \fig{logical-op-circuit}.
Its action is to first measure a set of $l$ independent commuting logical Pauli operators for the input code, and then apply an isometry from the resulting $(k-l)$-qubit code space to the $k'$-qubit output code space.
By analyzing the independent logical failure modes of the circuit, we see that $A$ has $K = k + k'$ rows.

\begin{figure}[ht]
  \centering
    (a)\includegraphics[height=2.4cm]{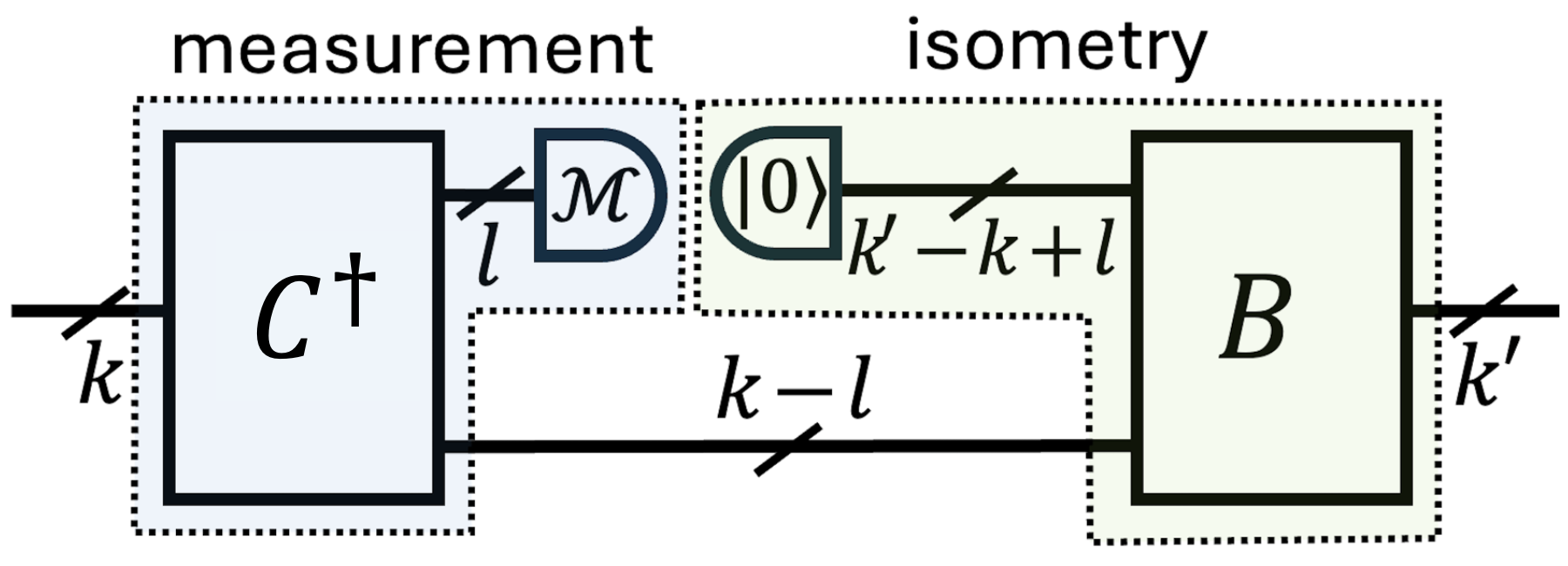}
    (b)\includegraphics[height=1.8cm]{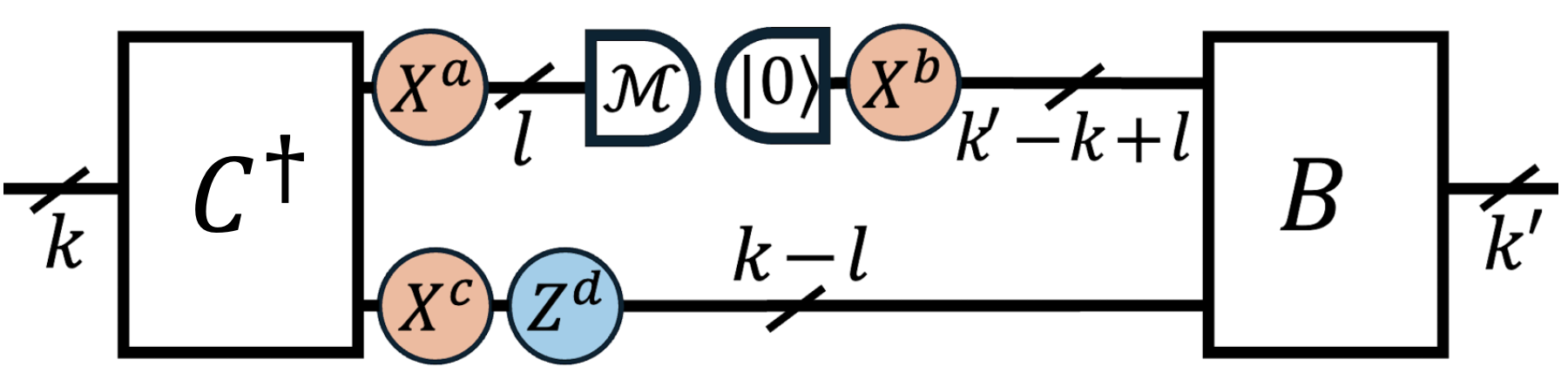}
     \caption{
    (a) The logical action of any stabilizer channel from $k$ to $k'$ logical qubits is to measure a set of $l$ independent commuting logical Paulis of the input code, and to apply an isometry from the remaining $k-l$ to $k'$ logical qubits of the output code.
    This is fully specified by $k$-dimensional and $k'$-dimensional Clifford unitaries $C$ and $B$, and the number $l = 0, 1, 2,\dots k$.
    (b) We use orange and blue disks to indicate logical $X$- and $Z$-type faults in this logical circuit that can have non-trivial action (note that a $Z$ error after a $|0\rangle$ preparation is trivial).
    We can specify all logical failures by the bitstrings $a\in\mathbb{F}_2^{l}$, $b\in\mathbb{F}_2^{k'-k+l}$, $c\in\mathbb{F}_2^{k-l}$ and $d\in\mathbb{F}_2^{k-l}$.
    Each of the $k+k'$ rows of $A$ corresponds to a bit in one of these bitstrings.
     }
  \label{fig:logical-op-circuit}
\end{figure}

\textbf{Detectors and logical outcomes: }
The circuit will in general include many measurements, although these are not necessarily associated with repeated stabilizer measurements for a stabilizer code.
Each detector is formed from the parity (or the complement of the parity) of a set of measurement outcomes, with trivial value in the absence of faults.
Each logical measurement outcome of the circuit is formed from the parity (or the complement of the parity) of a set of measurement outcomes, with the value equal to the logical outcome in the absence of faults.

To identify the sets of circuit outcomes forming detectors and logical outcomes, consider a modified circuit where each operation occurs in a separate time step, creating a sequential structure. 
In a fault-free stabilizer circuit, each measurement outcome is either~\cite{beverland2024fault}: (i) uniformly random, (ii) fixed (determined by the parity of a set of earlier measurement outcomes), or (iii) dependent on the logical state of the input code (specified by the parity of a set of earlier outcomes and the logical measurement outcome). 
A detector corresponds to each type (ii) outcome and includes that outcome and the earlier circuit outcomes it depends on. 
Similarly, a logical outcome corresponds to each type (iii) measurement, including that outcome and the earlier circuit outcomes it depends on.

\textbf{Constructing decoding matrices: }
Each fault is a distinct event that can: (a) introduce a Pauli operator on qubits in the circuit at some time (which can be pulled through the circuit to form a residual Pauli on the output), (b) cause detector outcomes to flip, (c) cause some logical measurement outcome to flip. 
For $N$ faults, and $M$ detectors, the check matrix is formed as usual, with a column for each fault and a row for each detector, with $H_{ij}=1$ iff detector $i$ is flipped by fault $j$.
Similarly, each fault can result in logical failures as specified in \fig{logical-op-circuit} such that the logical action matrix has a column for each of the $N$ faults and a row for each of the $K$ logical failure modes, with $A_{ij}=1$ iff fault $j$ leads to failure mode $i$.

Note that the resulting check and logical action matrices $H$ and $A$ may not be in the most convenient form.
Linearly independent recombinations of rows can be used to reduce the column weight etc.

\clearpage

\section{Decoding formalism}
\label{sec:decoding}

Here we fix the notational conventions needed for the rest of this work.

\subsection{Basic definitions}
\label{sec:basic-defs}
The primary definitions we will work with are in \defn{decoder}.

\begin{dfn}[Faults and decoding]
\label{defn:decoder}
Consider a check matrix $H \in \{0,1 \}^{M \times N}$, a logical action matrix $A\in \{0,1 \}^{K \times N}$, a probability vector $p \in (0,1/2)^{N}$ and a weights vector $w \in (0,\infty)^{N}$. 

Let the fault bitstring $F \in \{0,1\}^N$ be randomly drawn according to the distribution
$$ \Prob(F) =   \prod_{j =1}^N (1-p_j) \left(\frac{p_j}{1-p_j}\right)^{F_j}.$$
Given $F$, the syndrome $\sigma \in \{0,1 \}^{M}$ is obtained from $\sigma = H F$.
A \emph{decoder} is a classical algorithm which, given $\sigma \in \{0,1 \}^{M}$ (and with knowledge of the objects $H$, $A$ and $w$), proposes a correction $\hat{F}\in \{0,1\}^N$.
We say that the decoder \emph{succeeds} if both $H \hat{F}=\sigma$ and $A\hat{F} = AF$, and that it \emph{fails} otherwise. 
\end{dfn}
Arithmetic involving binary objects such as $\sigma = H F$ etc., are modulo two here and elsewhere.

The decoding problem phrased in \defn{decoder} is very general and applies to many settings of interest.
The objects $H,A,p,w$ here can be specified for a classical code, a quantum stabilizer code, a fault-tolerant stabilizer circuit that implements a stabilizer code, or something even more general such as a space time code. 
We include a discussion of these cases that explains how the objects arise in each case in \sec{decoding-scenarios}.

{\bf Low Density Parity Checks: }
We will assume that the check matrix $H$ is sparse, and more specifically, that it has a maximum column weight $c$ and maximum row weight $r$.
In other words, we are focused on decoding for quantum Low Density Parity Check (qLDPC) codes.

{\bf Weights: } 
Instead of providing the decoder with the probability vector $p$, we instead provide it with a \emph{weight} vector $w \in (0, \infty)^N$. 
For a fault bitstring $F$ we define its weight by
$$w(F) = \sum_{j=1}^N w_j F_j. $$
Unless otherwise stated, we will take the weight of fault $j$ to be given by the log-probability ratio (LPR), i.e., $w_j = \log \frac{1-p_j}{p_j}$, in which case $p$ and $w$ form a bijection and are thus equivalent.
We will make it clear when we find it convenient to assume a different choice of weight vector, such as \emph{uniform weights} where $w_j=1$ for all $j$ (in which case $w(F)$ is the Hamming weight of $F$).

{\bf Min-weight decoding: } 
A sub-optimal, but common strategy which often performs well in practice is where the fault bitstring $F$ with syndrome $\sigma$ with the minimum weight $w(F)$ is output, i.e.,
$$\hat{F}_\text{min-wt} = \mathop{\arg \min}_{F|HF=\sigma} (w(F)).$$
Note that (when weights are LPRs) finding an $F$ with minimal $w(F)$ is equivalent to finding an $F$ with maximum $\Prob(F)$, such that min-weight decoding is often called most-likely error (MLE) decoding.
For general stabilizer codes, min-weight decoding is NP-hard~\cite{hsieh2011}.

{\bf Optimal decoding: } 
The optimal decoding strategy, called most-likely coset (MLC) decoding, is to find a fault bitstring with syndrome $\sigma$ which maximizes the conditional probability of the coset of equivalent fault bitstrings, i.e.,
$$\hat{F}_\text{opt} = \mathop{\arg \max}_{F|HF=\sigma} \left(\sum_{S | HS =0, AS=0}\Prob(F + S) \right).$$
For general stabilizer codes, optimal decoding is \#P complete~\cite{iyer2015} (even harder than NP-hard).
We know of no results on the computational difficulty of general quantum low density parity check (LDPC) codes (where the check matrix has bounded row and column weight).

{\bf Syndrome height: } 
It is useful for some of our decoding strategies to define the \emph{syndrome height} $h(\sigma)$ as the weight of the lowest-weight fault configuration $F$ which has syndrome $\sigma$, i.e.,
$$ h(\sigma) = \min_{F | HF=\sigma} w(F). $$
Unless otherwise stated, the syndrome height will be computed using uniform weights.

\textbf{Decimation:}
It can be a useful step in some decoders to assume a subset of fault bits are included in a correction and then to consider the residual decoding problem that results.
Consider a check matrix $H$, a probability vector $p$ and a syndrome $\sigma$.
Given any fault bitstring $F$, we can define:
\begin{itemize}
    \item the \emph{decimated check matrix} $H_F$ is obtained by removing each $j$th column from $H$ where $F_j=1$,
    \item the \emph{decimated probability vector} $p_F$ is obtained by removing each $j$th element from $p$ where $F_j=1$.
\end{itemize}

\subsection{Decoding graph and set representation of bitstrings}

We often find it convenient to work with a graphical representation of $H$ (see \fig{example-decoding-graphs}).
The \emph{decoding graph}\footnote{One may wonder why we do not call $G$ the Tanner graph -- the reason is that in places we will discuss the same code under different noise models and we reserve the term Tanner graph for the defining representation of the code rather than having it depend on the noise model. For stabilizer codes, the decoding graph is the Tanner graph. 
This is discussed further in \sec{decoding-scenarios} and \sec{noise}.} $G$ is the bipartite graph which has biadjacency matrix $H \in \{0,1 \}^{M \times N}$.
There are two types of vertices: fault vertices (corresponding to columns of $H$), and check vertices (corresponding to rows of $H$).
We write $\calN(v)$ to specify the neighbors of a vertex $v$ in $G$.
We write $\calN(V) = \calN(v_1) \cup \calN(v_2) \cup \dots$ for a set of vertices $V = \{v_1,v_2,\dots \}$.
The \emph{fault degree} $c$ (\emph{check degree} $r$) is the maximum weight of any individual column (row) of $H$.

\begin{figure}[ht]
  \centering
    (a)\includegraphics[height=3.5cm]{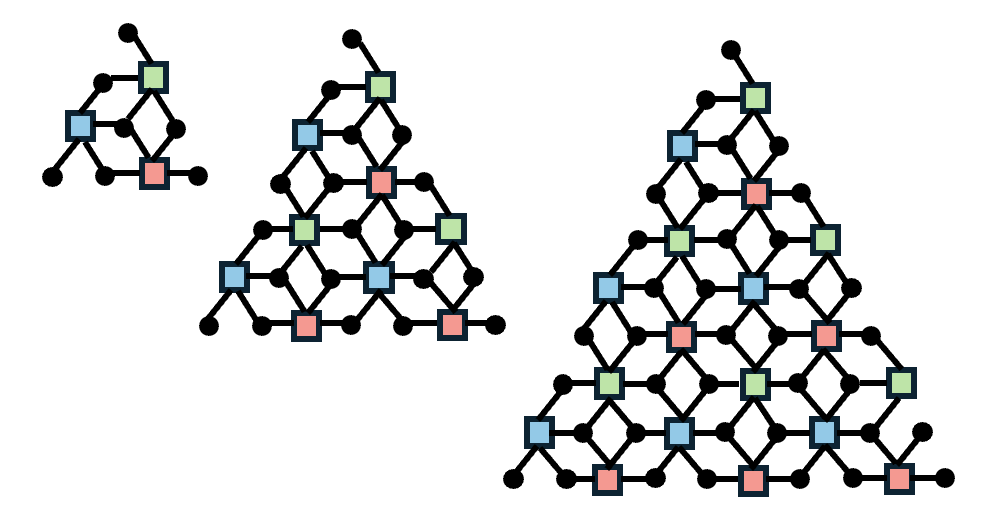}
    (b)\includegraphics[height=4cm]{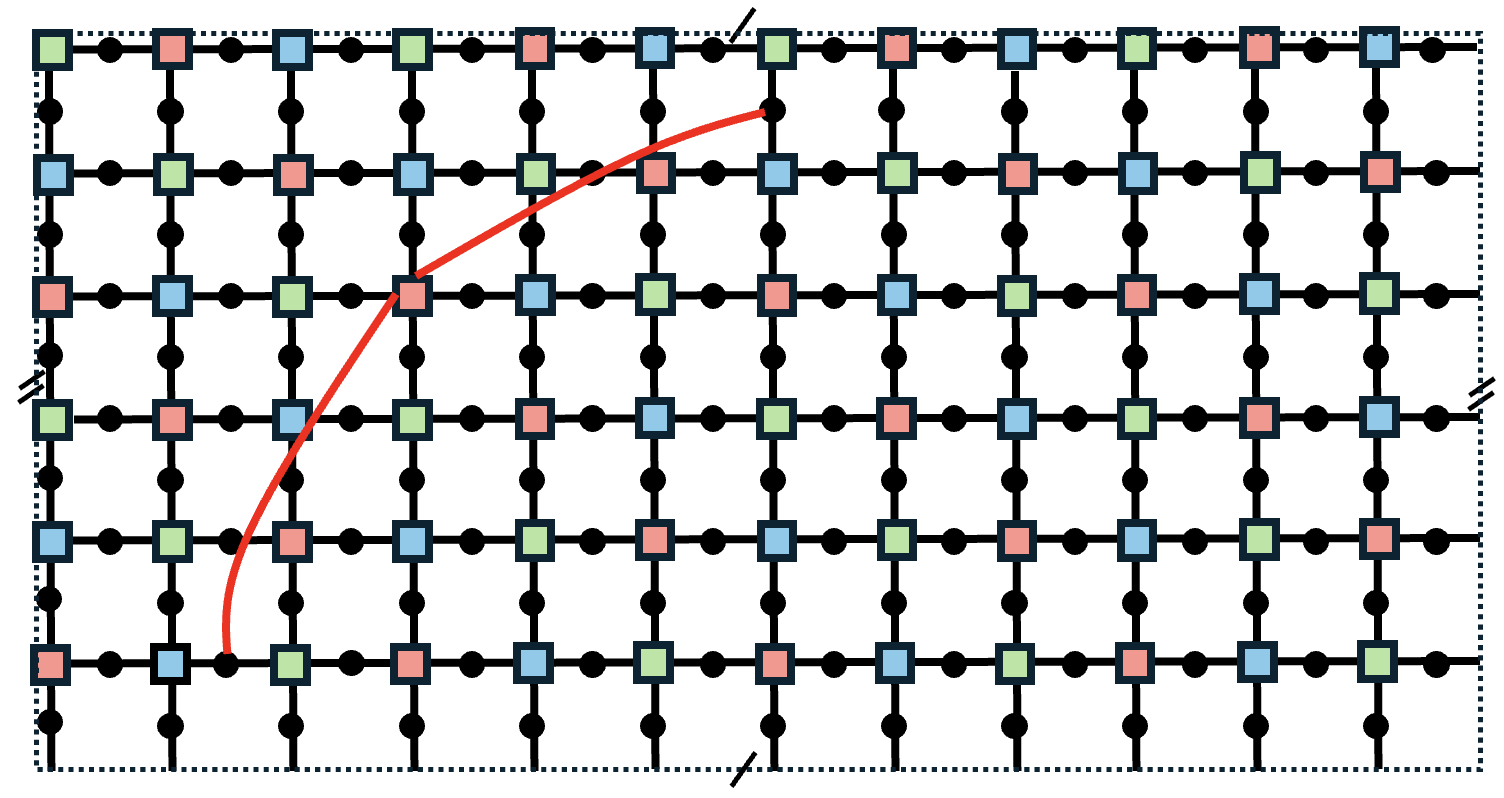}
     \caption{
     Decoding graphs for $X$-type noise with perfect measurements.
     There is a colored square vertex for each check, and a circle vertex for each fault (single-qubit $X$ error).
     (a) For three instances of the color code family with distances 3, 5 and 7 (larger distances can be obtained by extending the pattern).
     (b) For the gross code, where we draw only one pair of non-local edges (red), but there is a pair (not drawn) for for all square check vertices in the graph.
     See \fig{example-codes} for definitions of these codes.
     }
  \label{fig:example-decoding-graphs}
\end{figure}

We use $j = 1,2,\dots N$ to denote a fault vertex.
We then use $F$ to denote a set of fault vertices in two representations: as a set $F = \{j_1,j_2, \dots \}$, and as bit string $F \in \{0,1 \}^N$ (with $F_j = 1$ for $j \in \{j_1,j_2, \dots \}$).
Similarly, we use $\sigma$ to denote a set of check vertices both as a set $\sigma = \{i_1,i_2,\dots \}$ and as a bit string representation $\sigma \in \{0,1 \}^M$ (with $\sigma_i = 1$ for $i \in \{i_1,i_2, \dots \}$).
When $B$ and $C$ are interpreted as sets, we use $B+C$ to denote their symmetric difference, consistent with $B+C$ being modulo-two addition when $B$ and $C$ are interpreted as bit strings.
When $B$ and $C$ are interpreted as bit strings, we use $B \cup C$ to denote the (non-exclusive) or, consistent with $B\cup C$ being the union when $B$ and $C$ are interpreted as sets. 
Lastly, for $B$ and $C$ interpreted as sets, when $C \subseteq B$ we write $C-B$ to specify set difference for clarity (though this is equivalent to writing $B+C$).
The syndrome $\sigma$ of a fault set $F = \{j_1,j_2, \dots \}$ is $\sigma = \calN(j_1) + \calN(j_2) + \dots = \calN(F) =H F$.

\subsection{Belief propagation decoding}
\label{sec:decoding-algorithms}

Here we review the belief propagation (BP) decoding algorithm.
We use BP as a subroutine of some of the decoders we present in this work.

BP is a message-passing algorithm on the decoding graph which was originally designed for classical codes.
There are many versions of BP, we use a relatively standard version which was used for quantum codes in~\cite{panteleev2021degenerate, roffe2020}. 
BP can be thought of as an algorithm which, given the observed syndrome $\sigma$, estimates the posterior probability $\Prob(F_j=1|\sigma)$ that the $j$th fault occurred. 
The algorithm obtains this estimate by passing messages back and forth between fault vertices and check vertices over a sequence of iterations.

More specifically, the message passing is initialized by each fault vertex $j$ passing neighboring checks $i \in \calN(j)$ messages consisting of its LPR:
$$\mu^{(0)}_{j \rightarrow i} = \log\frac{1-p_j}{p_j}. $$
This is followed by alternating rounds of \emph{check-to-fault message} updates according to the minimum-sum rule 
$$\mu^{(t)}_{i \rightarrow j} =(-1)^{\sigma_i} \prod_{j' \in \calN(i) \setminus \{j\}} \mathrm{sign}(\mu^{(t-1)}_{j' \rightarrow i}) \cdot \min_{j' \in \calN(i) \setminus \{j\}} \card{\mu^{(t-1)}_{j' \rightarrow i}},$$
corresponding LPR updates
$$ \Lambda_j^{(t)} = \log\frac{1-p_j}{p_j} + \sum_{i \in \calN(j)} \mu^{(t)}_{i \rightarrow j},$$
and \emph{fault-to-check message} updates
$$\mu^{(t)}_{j \rightarrow i} = \Lambda^{(t)}_j - \mu^{(t)}_{i \rightarrow j}.$$
At every timestep $t$, the error pattern $\Hat{F}$ is estimated from the posterior log-probability ratios with $$\Hat{F}_j =\begin{cases} 0 \quad \mathrm{if} \quad \Lambda^{(t)}_j \geq 0 \\
1 \quad \mathrm{if} \quad \Lambda^{(t)}_j < 0 \end{cases} $$
The algorithm stops if the estimated error $\Hat{F}$ is consistent with the syndrome $\sigma=H \Hat{F}$, in which case we say BP converged, or if the maximum number of iterations is reached without providing a valid solution, in which case we say BP did not converge.

We denote the final LPR estimate of fault $j$ output by the last iteration as $\Lambda_\text{BP}(j) = \Lambda^{(t_\text{end})}_j$.

We also denote the arithmetic mean of the LPRs 
over the last $l_\text{buff}$ BP rounds as
\begin{eqnarray*}
\bar{\Lambda}_j = \sum_{l=0}^{l_\text{buff}-1} \Lambda^{(t_\text{end} - l)}_j.
\end{eqnarray*}
This can smooth possible oscillations occurring during BP similar to \cite{gong2024}.

Finally, we define the decimated LPR estimate $\Lambda_\text{BP}(j \,| \, \sigma, F)$ which is the $\Lambda_\text{BP}(j)$ obtained when running BP to decode $\sigma$ on the decimated check matrix $H_F$, and accordingly $\bar{\Lambda}_\text{BP}(j \,| \, \sigma, F)$.

A well-known disadvantage of BP is that it sometimes does not converge, especially if the decoding graph $G$ contains loops or if there are equivalent corrections with the same weight.
A common strategy is to use a follow-up decoder to handle these cases, with ordered statistics decoding being among the most commonly used option. 
We use the resulting two-step decoder known as BP-OSD to compare some of our decoders, and we review BP-OSD briefly in \app{bp-osd}.

\subsection{MaxSAT decoding}
\label{sec:maxsat-intro}

Here we briefly review the use of reductions to MaxSAT to form general qLDPC decoders~\cite{Berent_2024, noormandipour2024maxsatdecodersarbitrarycss}.

The MaxSAT problem involves finding a boolean variable assignment $b_l$ ($l = 1, 2, 3, \dots$) that satisfies the maximum number of clauses. 
Each clause is defined in terms of boolean variables using logical operations such as logical OR ($\lor$), logical XOR ($\oplus$), and logical NOT ($\bar{b}$, the negation of a boolean variable $b$). 
For example, typical clauses might include $(b_1 \lor \bar{b_2})$ and $((\bar{b_1} \oplus b_2) \lor b_3)$. 
By assigning weights to clauses, the MaxSAT problem generalizes to finding the assignment that minimizes the weight of unsatisfied clauses. 
For decoding, it is useful to distinguish between hard clauses (with infinite weight), which must be satisfied, and soft clauses (with finite weight), where minimizing the total weight of unsatisfied clauses is the objective.

The minimum-weight decoding problem can be formulated as a MaxSAT problem, where the boolean variables are the bits of the correction $F_j$.
Hard clauses are derived from the syndrome equation $\sigma = HF \mod 2$, expressed as:
\[
\sigma_i \oplus \left(\bigoplus_{j=1}^N H_{ij} F_j\right), \quad i = 1, \dots, M.
\]
Soft clauses arise from minimizing the weight $w(F)$, introducing $N$ clauses:
\[
\bar{F_j}, \quad \text{with weight } w_j, \quad j = 1, \dots, N.
\]    

Constructing an end-to-end decoder via MaxSAT reduction requires solving the MaxSAT instance. Ref.~\cite{Berent_2024} used Z3~\cite{10.1007/978-3-540-78800-3_24}, while Ref.~\cite{noormandipour2024maxsatdecodersarbitrarycss} found Open-WBO~\cite{10.1007/978-3-319-09284-3_33} most effective. 
Practical solvers often perform better with reformulated clauses (different, but logically equivalent to those stated above), such as converting the problem into Max 3-SAT in conjunctive normal form using auxiliary variables, as described in \cite{noormandipour2024maxsatdecodersarbitrarycss}.  
MaxSAT decoding achieves significantly lower logical error rates than heuristic methods like BP-OSD but is substantially slower, even with state-of-the-art solvers optimized for MaxSAT benchmarking competitions \cite{Berent_2024, noormandipour2024maxsatdecodersarbitrarycss}. 

One challenge of the MaxSAT mapping approach to qLDPC decoding is that it seems quite difficult to incorporate structure to improve the performance for specific code families, and also difficult to incorporate the notion of stabilizer equivalence that exists in decoding of quantum codes.
In \app{maxsat-comparison}, we compare the runtime of our minimum-weight decoder (height-bound DTD) with data for the MaxSAT decoder in \cite{Berent_2024, noormandipour2024maxsatdecodersarbitrarycss}.

\clearpage

\section{Decision-tree decoding}
\label{sec:main-algo}

In this section we first define the decision tree in \sec{decision-tree-object}, then introduce the main decision-tree decoding (DTD) \alg{decision-tree-decoder} in \sec{main-DTD-algo}. 
The main DTD algorithm is defined with respect to a subroutine which specifies how the decision tree is explored.
In \sec{breadth-first} and \sec{oracle} we introduce and analyze some conceptually simple exploration subroutines that serve as a warm-up to build intuition for the more practically relevant ones which we provide later in \sec{provable-cost} and \sec{fast-cost}.

\subsection{Decision tree}
\label{sec:decision-tree-object}

The \emph{decision tree} is the main object needed to understand the family of decision-tree decoders.
Here we reiterate and expand upon the description of the decision tree that we provided in \fig{intro-figure} of \sec{intro}. 

The decision tree $\decisionTree$ is defined for a particular syndrome $\sigma$ on the decoding graph $G$.
Start by including a root node for $\decisionTree$. 
Next, take a check vertex $i$ in the set $\sigma$ and add a child to the root node for each neighboring fault $j \in \calN(i)$. 
Each child (labeled by $(j)$ for $j \in \calN(i)$) in turn has its own children (labeled by $(j,j')$) for $j' \in \calN(i')\setminus \{ j \}$, where $i'$ is a check vertex in the set $\sigma+\calN(j)$, the \emph{updated syndrome} that would be obtained by applying the fault $j$ when the original syndrome was $\sigma$.
The full $\decisionTree$ is obtained by continuing this procedure iteratively until all branches terminate in a leaf. 
Each node in $\decisionTree$ is then uniquely labeled by the path of faults which were taken to reach it.

Note $\decisionTree$ is never fully constructed in decision-tree decoding, rather parts of it are generated procedurally, as we explore it.

\begin{figure}[ht]
  \centering
    \includegraphics[width=1.0\textwidth]{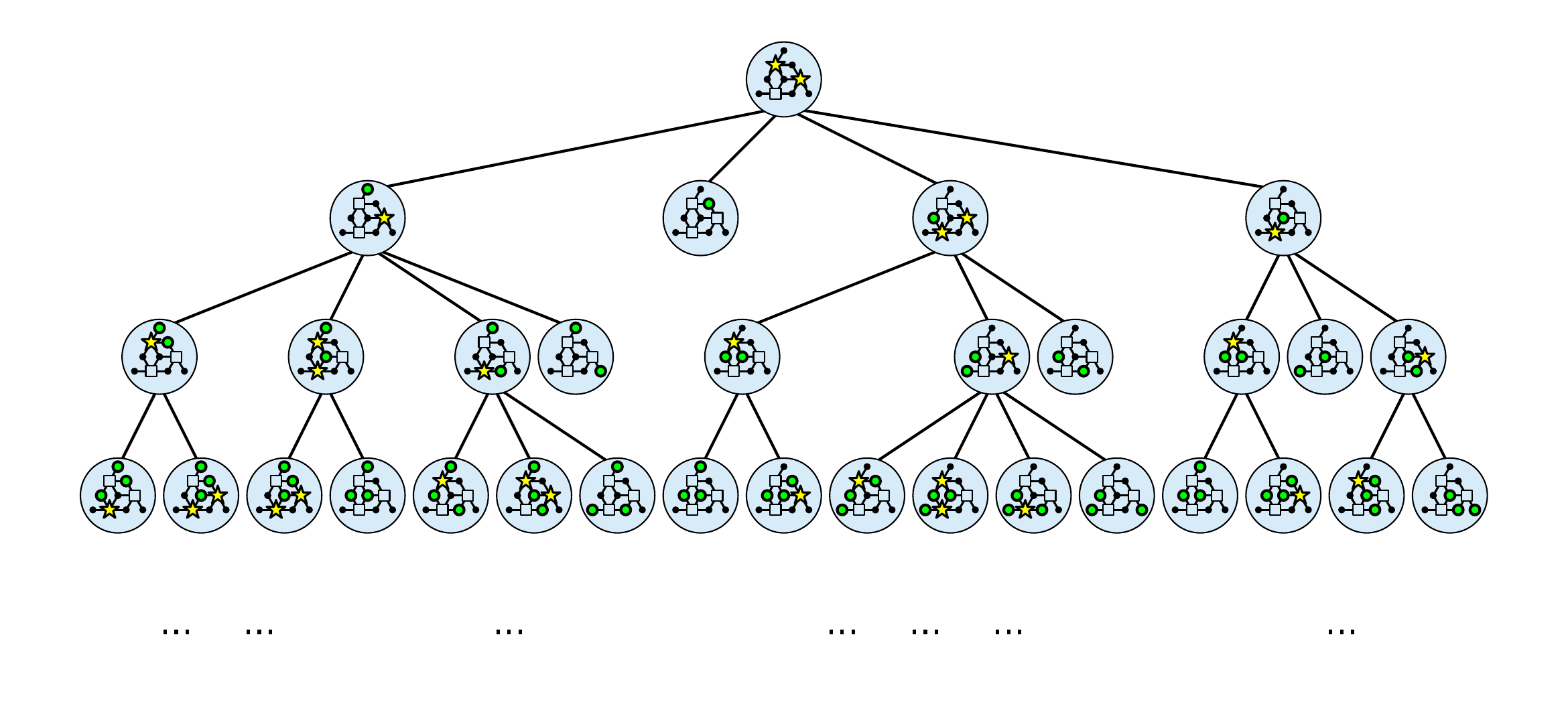}
     \caption{
     First 4 levels of the decision tree $T$ of the Steane code~\cite{steane1996multiple} for the syndrome given in the root node. 
     (We draw the Steane code as the $d = 3$ color code as in \fig{example-codes}(a,c)).
     In each node we show the updated syndrome (highlighted with yellow stars) and the decisions taken to reach it (highlighted green fault vertices), with the original syndrome $\sigma$ at the root node. 
     }
  \label{fig:decision-tree}
\end{figure}

It can be useful to visualize a decoding graph for each node in $\decisionTree$, where we highlight the set $F$ of fault vertices in the path to the node, and also the updated syndrome $\sigma+HF$ (see \fig{decision-tree})).
Each node in the decision tree can be considered as an attempted correction that could be applied, which succeeds if the updated syndrome is trivial. 
Every minimum weight correction corresponds to a leaf in the tree. 
Since only the set (and not the order-specific sequence) of faults $F$ in a path fully specify the applied correction, and also the sub tree starting from this node, $\decisionTree$ contains a lot of redundancy. \alg{decision-tree-decoder}, to be discussed later, avoids exploring these equivalent sub trees.

The decision tree is very large, but it is always finite for any given syndrome $\sigma$ and decoding graph $G$.
To see that all branches in the tree $\decisionTree$ terminate on a leaf, consider a node in the tree that has been reached by a path through a set of faults $F$, and note that any children of this node must correspond to a fault $j \in \calN(\sigma+HF)\setminus F$.
If $\calN(\sigma+HF)\setminus F$ is empty, then the node has no children and is therefore a leaf. 
Therefore, a path through the tree must terminate at the latest when $\card{F} = N$ and so the tree has at most depth $N$.
Note that every solution node (i.e., when $\sigma+HF = \emptyset$) is a leaf, but not every leaf corresponds to a solution (when $\sigma+HF \neq \emptyset$ but $\calN(\sigma+HF) \subseteq F$).
The longest path from the root to a leaf is then at most as long as the total number of fault nodes in $G$.
One final point to note is that, the decision tree is constructed by picking any check vertex of the updated syndrome at each step, and as such the tree depends on that choice.
The tree arising from any such choice can be used for the algorithm.

\subsection{Main algorithm and exploration subroutine}
\label{sec:main-DTD-algo}

Here we introduce the main decision tree decoding (DTD) \alg{decision-tree-decoder}.
At a high level, the algorithm searches for a correction for the syndrome $\sigma$ by exploring the decision tree $\decisionTree$, starting at the root. 
As the algorithm proceeds it dynamically generates or grows parts of $\decisionTree$ until a solution is found.
More specifically, it stores the set $\tree$ of all `seen' nodes, and separately stores a set $\leaves$  of `live' nodes, which retains the subset of seen nodes, which have not yet been explored. 

To guide the algorithm's exploration, each node in $\leaves$ is assigned a \emph{cost}.
In each round, the lowest-cost live node is selected and the algorithm stops if it provides a valid correction. 
Otherwise, we run an \texttt{Explore} subroutine, which 'explores' this node in the decision tree, that is it identifies the faults that lead to its children and assigns them a cost for further exploration. 
All variants of the decision-tree decoder family differ only by this \texttt{Explore} subroutine, which defines order the search space is explored through the assignment of cost. Additionally, the \texttt{Explore} subroutine can directly return a solution.
If $\leaves$ becomes empty at any point in the algorithm, without having produced a solution, the algorithm terminates indicating that no correction exists that is consistent with the syndrome.

\begin{algorithm}
    \caption{Decision-tree decoder (main)}
    \label{alg:decision-tree-decoder}
    \begin{algorithmic}[1] 
        \Procedure{Decode}{$\sigma_\text{in}$} \Comment{Decode syndrome $\sigma_\text{in}$}
            \State $\tree \gets \{ \emptyset \}$ \Comment{Initialize seen tree nodes} \label{lst:line:init1}
            \State $\leaves \gets [(\emptyset, \sigma_\text{in}, \zeroc)]$ \label{lst:line:init2}  \Comment{Initialize live nodes (fault set, syndrome, cost)}
            \While{$\leaves$ \textbf{not} empty}
                \State ($F, \sigma,  C) \gets \texttt{first}(\leaves)$ \label{lst:line:extract} \Comment{Extract cheapest node}
                \If{$\sigma$ empty} \label{lst:line:start_check} 
                    \State \textbf{return} $F$ \Comment{Finish if correction found}
                \EndIf \label{lst:line:end_check}

                \Statex \hspace{\algorithmicindent} \hspace{8pt} \hrulefill 
                \State \textbf{Exploration subroutine:} \label{lst:line:exploration_block}
                \Statex \hspace{\algorithmicindent} \hspace{8pt} \textit{Identifies the fault vertices $j_l$ leading to the current node's children. }
                \Statex \hspace{\algorithmicindent} \hspace{8pt} \textit{Assigns them a cost $C^{(j_l)}$.}
                \Statex \hspace{\algorithmicindent} \hspace{8pt} \textit{\textbf{Computes:} $\{(j_1,C^{(j_1)}),\dots, (j_b, C^{(j_b)}) \}$}
                \Statex \hspace{\algorithmicindent}   \hspace{8pt} \hrulefill 
                
                \For{$j \in \{j_1,\dots, j_b\}$ }  \label{lst:line:start_grow}
                    \State $F^{(j)} \gets F \cup \{ j \}$ \Comment{Updated fault set}
                    \State $\sigma^{(j)} \gets \sigma + \calN(j)$ \Comment{Updated syndrome}
                    \If{$F^{(j)} \notin \tree$} \label{lst:line:pruning}
                        \State insert $(F^{(j)},\sigma^{(j)},C^{(j)})$ into $\leaves$
                        \State add $F^{(j)}$ to set $\tree$
                    \EndIf
                \EndFor \label{lst:line:end_grow}
            \EndWhile
            \State \textbf{return} failure \Comment{No correction was found}
        \EndProcedure
    \end{algorithmic}
\end{algorithm}

Let us add a few comments to clarify a number of aspects of \alg{decision-tree-decoder}:
\begin{itemize}
    \item This `main' algorithm applies for all decision tree decoders that we consider in this work.
    The difference between different versions of the algorithm is captured entirely within the inlined \texttt{Explore} subroutine, which controls how the decision tree is explored by assigning costs to different partial corrections. 
    Different costs produce very different DTD algorithms. 
    We will generally name specific DTD algorithms after their exploration subroutines.
    
    \item No significant part of $\decisionTree$ is ever fully constructed and stored in \alg{decision-tree-decoder}, instead we explore small parts of it and store (in $\tree$) relevant information about what has been explored, and store (in $\leaves$) information about what could be considered for further exploration. 

    \item Since the correction $F$ (and the corrections of all descendants) found by the exploration of $\decisionTree$ does not depend on the order faults were added to $F$ by the algorithm, we just store $F$ in $\tree$. 
    Fault sets which are considered for exploration are checked against those already seen to avoid redundancy.
    This ensures \alg{decision-tree-decoder} never explores a node of $\decisionTree$ if an equivalent node has already been explored.
    
    \item We store `live nodes' which may be further explored in $\leaves$.
    Each element $(F, \sigma, C)$ consists of the partial correction $F$, the updates syndrome $\sigma$ and a cost $C$.
    We always explore the cheapest element, so when implementing this algorithm in practice it makes sense to store elements in $\leaves$ according to their cost (for example as a heap).
    
    \item Many objects, including partial corrections $F$ and updated syndromes $\sigma$ are expected to be sparse bit strings (or equivalently small sets), which makes sense to take advantage of in implementations. 
\end{itemize}

To compare the efficiency of decision tree decoders resulting from different exploration subroutines in \alg{decision-tree-decoder}, we define the \emph{number of explored nodes} $\nu$ as the total number of nodes the exploration subroutine is called on before finding a correction.

In the following subsections (\sec{breadth-first} and \sec{oracle}) we consider the tree-exploration approaches generated by some simple but illuminating cost functions, which illustrates how very different behavior of decision-tree decoding algorithms can arise.

\subsection{Exploration using breadth-first search}
\label{sec:breadth-first}

Here we consider the first of our two warm-up examples of DTD algorithms (which are illustrative but impractical).
A minimum-weight correction corresponds to the shortest path from the root to a leaf in $\decisionTree$, or more generally, the path with the lowest total weight when weights are non-uniform.  
A standard method for finding such a path is weighted breadth-first search, making it a natural approach for minimum-weight decoding.

When the \texttt{Explore} \sub{bf-cost} is used in \alg{decision-tree-decoder}, we call the resulting decoder the \emph{Breadth-first DTD}.
In the \texttt{Explore} \sub{bf-cost}, assigning the cost change from parent to child as the weight of the selected fault vertex naturally results in a weighted breadth-first exploration of the decision tree, yielding a minimum-weight correction (see \lem{breadth-first-exploration}).

\begin{subroutine}
    \caption{(Weighted) breadth-first exploration}
    \label{alg:bf-cost}
    \begin{algorithmic}[1]
        \State \textbf{pick} $i \in \sigma$ 
        \For{$j  \in  \calN(i) \! \setminus \! F = \{j_1, j_2, \dots j_b \}$}
            \State $C^{(j)} \gets C+w_j$
        \EndFor
    \end{algorithmic}
\end{subroutine}

\begin{lemma}[Breadth-first DTD: min-weight]
\label{lem:breadth-first-exploration}
When using exploration \sub{bf-cost}, the DTD \alg{decision-tree-decoder} returns a minimum-weight correction.
\end{lemma}
\begin{proof}
The algorithm explores all nodes of the decision tree in order of increasing weight until it finds a correction.
All minimum-weight corrections are contained in the decision tree, and as such the first correction that is found will be of minimum weight and is output by the decoder.
\end{proof}

\subsection{Exploration using syndrome height (from an oracle)}
\label{sec:oracle}

Here we consider the second of our two warm-up examples of DTD algorithms (which are illustrative but impractical).
The breadth-first exploration discussed in \sec{breadth-first} is guaranteed to find a minimum-weight correction but is expected to be very slow as it will need to look at all lower-weight sets of faults first before finding the solution. 
Here we consider a different cost function, namely the syndrome height $\synht(\sigma)$ (which is the weight of a minimum-weight correction for the syndrome $\sigma$; see \sec{basic-defs}).
The syndrome-height cost function results in an exploration that is guaranteed to find a minimum-weight correction for $\sigma$ applying the explore subroutine to $O(h(\sigma))$ corrections in the decision tree.
There is a catch: there is no efficient way to compute the syndrome height in general, and so we cannot form a practical decoder from this.
However we find it conceptually useful to understand how the decision-tree decoder would work assuming access to a hypothetical oracle for the syndrome height.

When the \texttt{Explore} \sub{oracle-cost} is used in \alg{decision-tree-decoder}, we call the resulting decoder the \emph{Height-oracle DTD}.
The exploration subroutine based on syndrome height is given in \sub{oracle-cost}.

\begin{subroutine}
    \caption{Height oracle exploration}
    \label{alg:oracle-cost}
    \begin{algorithmic}[1] 
        \State \textbf{pick} $i \in \sigma$
        \For{$j  \in  \calN(i)\! \setminus \! F = \{j_1, j_2, \dots j_b \} $}
            \State $\sigma^{(j)} \gets \sigma + \calN(j)$
            \State $C^{(j)} \gets \synht(\sigma^{(j)})$ \Comment{From $\texttt{HeightOracle}$}
        \EndFor     
    \end{algorithmic}
\end{subroutine}

\begin{lemma}[Height-oracle DTD: min-weight, min explored nodes]
\label{lem:efficient-oracle-exploration}
When using exploration \sub{oracle-cost}, the DTD \alg{decision-tree-decoder} returns a minimum-weight correction $\hat{F}^*$ after exploring a minimum number of nodes $\nu = |\hat{F}^*|$.
\end{lemma}
\begin{proof}
Consider a tree node with syndrome $\sigma'$ and a child node with syndrome $\sigma'' = \sigma' + \calN(j)$ for some fault node $j$. 
Then $\synht(\sigma'') \geq \synht(\sigma') - w_j$ with equality attained iff $j$ is in a minimum-weight correction for $\sigma'$.  
Always going down a branch that maximally decrements the cost thus directly leads to a minimum-weight correction $\hat{F}^*$ for $\sigma$ in $\card{\hat{F}^*}$ calls of the explore subroutine.
\end{proof}

This height-oracle exploration is illustrated in \fig{intro-figure} in \sec{intro}.
While we know of no efficient way to compute the syndrome height (we believe there is none), using syndrome height as a cost measure inspires the other more practical algorithms that we introduce in \sec{provable-cost} and \sec{fast-cost}, where the syndrome height is efficiently bounded or estimated.

\clearpage

\section{Height-bound decision-tree decoder}
\label{sec:provable-cost}

In this section, we introduce a decoder that explores the decision tree using a cost function based on lower bounds of the syndrome height, and which is guaranteed to find a minimum-weight correction.
The catch is that this decoder is not guaranteed to terminate quickly for all error patterns (it works best when syndrome-height lower bounds are relatively tight).
The intuition behind this approach is that lower bounds of the syndrome height allow large branches of the decision tree to be avoided knowing that they cannot lead to a lower weight error than that which is ultimately found.

We introduce the basic exploration subroutine in \sec{exploration-height-bounds}, and prove it provides a minimum-weight correction.
We then introduce a refined version of the exploration in \sec{cost-refinement} which does not affect the minimum-weight guarantee, but which uses BP to help find a solution more quickly.
This is the exploration version we use in our \emph{height-bound} DTD.
We characterize the runtime of the height-bound DTD numerically in \sec{min-weight-numerics}, observing that the median-case runtime scales linearly with the weight of the error for a range of 2D color codes and bivariate bicycle codes.
There are also two notable appendices associated with the material in this section: 
In \app{maxsat-comparison} we compare the runtime of height-bound DTD with published data for MaxSAT decoders, which are also min-weight.
In \app{find-min-weight-logicals} we provide an algorithm that uses the decision tree and the height bound to compute all min-weight logical operators of a code.

\subsection{Exploration using syndrome-height lower bounds}
\label{sec:exploration-height-bounds}
Recall that the decision tree exploration in \sec{oracle} based on syndrome height finds a minimum-weight correction in linear time, but that we know of no efficient way of computing the syndrome height.
In that algorithm, the syndrome height is useful because it can be used to identify the weight of the minimum-weight correction contained in the descendants of any node in the decision tree, eliminating all but the optimal node for exploration during each iteration.
Later in \sec{finding-lower-bounds} we will see that while we cannot efficiently compute the syndrome height, we have lots of techniques to find lower bounds.
While this does not eliminate all but the optimal node during each round, it can eliminate many nodes.

To see how this works, suppose that we have reached a node labeled by the fault set $F$ with updated syndrome $\sigma$, and suppose that $\synht_\text{min}(\sigma)$ is a lower bound of $\synht(\sigma)$.
Let the cost of node $F$ be
\begin{equation}
\label{eq:minweight-cost-def}
    C = w(F) + h_\text{min}(\sigma).
\end{equation}
Note that $C \leq w(F) + \synht(\sigma)$, i.e., $C$ lower bounds the weight of any correction that could be reached by exploring node $F$ and its descendants. 
As such, if we find a correction by exploring some \emph{other} node with cost less than or equal to $C$, then we know that we will not be able to find anything better by exploring node $F$.

\begin{subroutine}[ht]
    \caption{Height-bound exploration (unrefined version)}
    \label{alg:minweight-cost}
    \begin{algorithmic}[1] 
        \State \textbf{pick} $i \in \sigma$
        \For{$j  \in  \calN(i) \! \setminus \! F = \{j_1, j_2, \dots j_b \}$}
            \State $F^{(j)} \gets F \cup \{j\}$
            \State $\sigma^{(j)} \gets \sigma + \calN(j)$
            \State $h_\text{min} \gets \texttt{height\_lower\_bound}(\sigma^{(j)})$
            \State $C^{(j)} \gets \max \left[h_\text{min} + w(F^{(j)}) ,C\right]$  \Comment{Use bound from $h_\text{min}$ unless bound from parent stronger}
        \EndFor
    \end{algorithmic}
\end{subroutine}

In \sub{minweight-cost} we show the exploration subroutine that inserts the cost of \eq{minweight-cost-def} into the DTD decoder \alg{decision-tree-decoder} by calling a function $\texttt{height\_lower\_bound}$.
We call this the (unrefined version) of the height-bound exploration to distinguish it from that which we use for our main height-bound DTD which is provided in \sec{cost-refinement}.
When the \texttt{Explore} \sub{minweight-cost} is used in \alg{decision-tree-decoder}, we call the resulting decoder the \emph{(unrefined) height-bound DTD}.

\begin{lemma}[Height-bound DTD: min-weight]
\label{lem:height-bound-lemma}
When using exploration \sub{minweight-cost}, the DTD \alg{decision-tree-decoder} returns a minimum-weight correction.
Similarly, using exploration \sub{minweight-enhanced-cost}, the DTD \alg{decision-tree-decoder} returns a minimum-weight correction.
\end{lemma}
We prove the case for \sub{minweight-cost} (and later explain why the proof still holds for \sub{minweight-enhanced-cost}).
\begin{proof}
    Consider three nodes (see \fig{wlb_dt}):
    \begin{enumerate}
        \item $(\hat{F}^*, \emptyset, C^*)$: The solution found by \alg{decision-tree-decoder} (obtained in line \ref{lst:line:extract} from $\leaves$).
        \item $(\hat{F}, \emptyset, C)$: Any other arbitrary solution in the decision tree.
        \item $(\tilde{F}, \tilde{\sigma}, \tilde{C})$: The ancestor node of $(\hat{F}, \emptyset, C)$ currently in $\leaves$.
    \end{enumerate}
    We show that the found solution $\hat{F}^*$ has minimum weight, i.e. that $w(\hat{F}^*) \leq w(\hat{F})$, as follows:
    \begin{align}
        w(\hat{F}^*) &= C^* - h_\text{min}(\emptyset) &&  \text{definition of cost} \nonumber \\
             &= C^* && \text{as } 0 \leq h_\text{min}(\emptyset) \leq \synht(\emptyset) = 0 \label{eq:solcostiscorrweight} \\
           &\leq \tilde{C} && \text{definition of algorithm } \nonumber \\
           &\leq w(\tilde{F}) + \synht(\tilde\sigma) && \text{definition of cost}\nonumber \\
           &\leq w(\tilde{F}) + w(\hat{F}-\tilde{F}) && \text{as } H( \hat{F}-\tilde{F}) = \tilde{\sigma} \nonumber\\
           &= w(\hat{F}) && \text{by linearity of weight}\nonumber
    \end{align}
    as required.
\end{proof}

\begin{figure}[ht]
  \centering
    \includegraphics[width=0.5\textwidth]{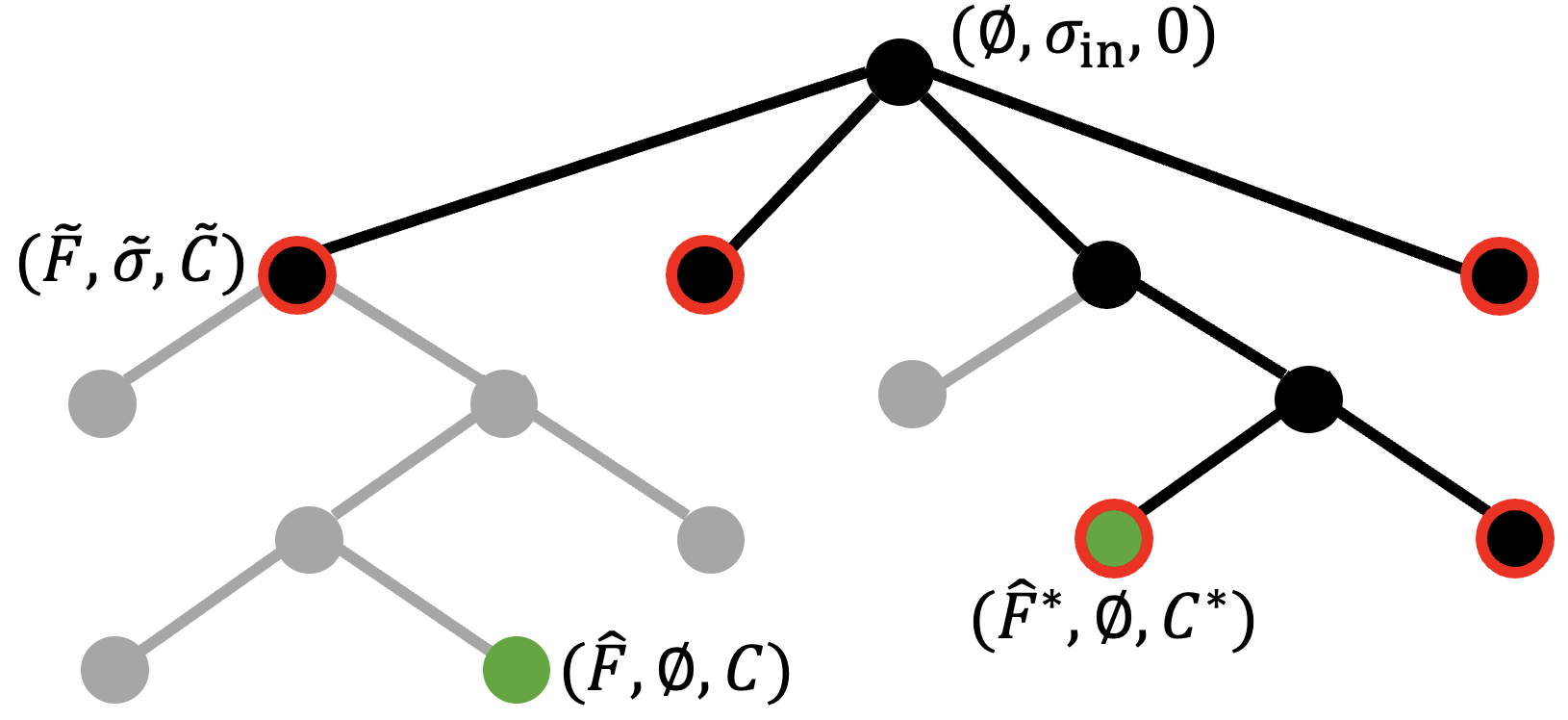}
     \caption{
     Decision tree decoding with the cost function $C = w(F) + h_\text{min}(\sigma)$ yields a minimum weight solution.
     We sketch part of the decision tree $\decisionTree$ just before the the algorithm finds a correction $\hat{F}^*$.
     Nodes filled black have been seen, while nodes filled gray have not, and nodes outlined in red are live and under consideration for exploration. 
     Two nodes corresponding to valid corrections (i.e., have trivial updated syndrome) are filled green, including $\hat{F}^*$ and another node $\hat{F}$. 
     We also label the node $\tilde{F}$, which is the live ancestor of $\hat{F}$.
     The fact that correction $\hat{F}^*$ is found implies that $C^* = w(\hat{F}^*)$ must be the lowest cost for any live node, such that in particular $\tilde{C} \geq C^*$.
     Then $\hat{F}^*$ must be a minimum-weight solution as $w(\hat{F}^*) = C^* \leq \tilde{C} \leq w(\hat{F})$. 
     }
  \label{fig:wlb_dt}
\end{figure}

We make the following remarks:
\begin{itemize}
    \item The weighted breadth-first exploration \sub{bf-cost} is a special case of \sub{minweight-cost} with the trivial lower bound $h_\text{min}(\sigma) = 0 \quad \forall \, \sigma$.

    \item Since in decision-tree decoding we cannot select a fault twice, we could replace $\synht(\sigma)$  with the \emph{decimated syndrome height} $\synht_F(\sigma)$ (height of the syndrome $\sigma$ on the decimated check matrix $H_F$) in \eq{minweight-cost-def} and the proof would still hold. 
    Taking decimation into account would lead to stronger bounds but possibly increase the complexity of finding such bounds. 
    
\end{itemize}

\subsection{Exploration refinement with belief propagation}
\label{sec:cost-refinement}
The cost defined in \eq{minweight-cost-def} is discrete which can lead to ties that make for a rather arbitrary search (especially when uniform weights are used).
For any choice among the ties, the algorithm will eventually find a minimum weight solution, but some choices could do so much faster than others.

To remedy this, we run belief propagation and use its output as a tiebreaker. 
We generalize the cost to a vector $\bm{C}=(C_\text{bound}, C_\text{tie-break}) \in \mathbb{R}^2$ and cost vectors are compared by first comparing their $C_\text{bound}$ components (from the height lower bound), and then comparing $C_\text{tie-break}$ (for which we use BP) only if the former reports equality.
The tie-break cost of the parent node is initialized as $C_\text{tie-break} = 0$, 
and then we run belief propagation, using the estimated LPR output $\Lambda_\text{BP}(j \,| \, \sigma)$ as  $C^{(j)}_\text{tie-break}$ for each child (specified by fault vertex $j$).
This process continues: the tie-breaker cost of the $j$th child of a parent node with fault set $F$ is found by running decimated BP and adding the estimated LPR output $\Lambda_\text{BP}(j \,|\sigma,F)$ to the tie-breaker cost of the parent.
Since the LPR is a large negative value for a fault that BP is confident occurred, lower costs are preferred.

When the resulting \texttt{Explore} \sub{minweight-enhanced-cost} is used in \alg{decision-tree-decoder}, we call the resulting decoder \emph{height-bound DTD}.
Since we only invoke $C_\text{tie-break}$ for tie-breaking cases of $C_\text{bound}$, the proof of \lem{height-bound-lemma} remains unaffected by replacing \sub{minweight-cost} by \sub{minweight-enhanced-cost}.

\begin{subroutine}[ht]
    \caption{Height-bound exploration}
    \label{alg:minweight-enhanced-cost}
    \begin{algorithmic}[1] 
        \State $C_\text{bound}, C_\text{tie-break} \gets \bm{C}$
        \State $\set{\Lambda_\text{BP}(j|\sigma, F)}, \, \hat{F}_\text{BP} \gets \texttt{BPDecimated}(\sigma,F)$ \Comment{Run BP on $\sigma$ with decimation $F$}
        \State \textbf{pick} $i \in \sigma$
        \For{$j  \in  \calN(i) \! \setminus \! F = \{j_1, j_2, \dots j_b \}$}
            \State $F^{(j)} \gets F \cup \{j\}$
            \State $\sigma^{(j)} \gets \sigma + \calN(j)$
            \State $h_\text{min} \gets \texttt{height\_lower\_bound}(\sigma^{(j)})$
            \State $\bm{C}^{(j)} \gets \left(\max \left[h_\text{min} + w(F^{(j)}) ,C_\text{bound}\right],\,  C_\text{tie-break} + \Lambda_\text{BP}(j \,| \, \sigma, F) \right)$
        \EndFor
    \end{algorithmic}
\end{subroutine}

\subsection{Syndrome-height lower bounds from syndrome neighborhoods}
\label{sec:finding-lower-bounds}

We have not yet specified how to find lower bounds on the syndrome height $h(\sigma)$. 
We require bounds to be efficient to evaluate, and ideally they should be somewhat tight for typical syndromes, which can depend on the check matrix $H$ of the code. 
We assume uniform fault weights in this section.

In this section, we present bounds that can be applied to any qLDPC code based on the requirement that every syndrome vertex must have at least one fault in its neighborhood.
The resulting `syndrome-neighborhood bounds' are expected to be weak for topological codes, due to the existence of string operators with syndromes only at their endpoints. 
However, syndrome-neighborhood bounds ought to be tighter for codes with expansion, where large errors produce large syndromes. 
Nevertheless, as shown in \sec{min-weight-numerics}, height-bound DTD using syndrome-neighborhood bounds achieves excellent median-case runtime for 2D color codes, a topological code family (and we find even better performance for bivariate bicycle codes).

\textbf{Loose bound from fault degree: }
Suppose $G$ has fault degree $c$ (which is the maximum number of check vertices touching any individual fault vertex).
Given a syndrome $\sigma$, a single fault can be responsible for a maximum number of $c$ check vertices in $\sigma$ and so a height lower bound is \begin{eqnarray}
\label{eq:degree-bound}
\synht(\sigma) \geq h_1(\sigma) =  \Bigl\lceil \frac{\card{\sigma}}{c} \Bigr\rceil.
\end{eqnarray}

\textbf{Tighter bound from syndrome structure: }
Clearly the bound in \eq{degree-bound} can be strengthened in certain cases. 
For example, consider a scenario where no two check vertices in $\sigma$ neighbor the same fault vertex (which would imply $\synht(\sigma) \geq |\sigma|$).
To find a tighter bound, we need a more detailed analysis.
The following bound arises from the requirement that there must be at least one fault in the neighborhood of every syndrome vertex, making refinements based on the structure of overlaps of those neighborhoods. 

\begin{lemma}
Consider a syndrome $\sigma$ caused by an (unknown) correction $F$.
Let $c$ be the max column weight of $H$ (equivalently, the maximum number of check nodes adjacent to any fault node in $G$).
Let $B_l$ be the set of all fault vertices which touch precisely $l$ vertices in $\sigma$ for $l=1,2,\dots, c$.
(Note that $B_l$ depends on $\sigma$ and the decoding graph $G$ alone, not on $F$.)
Let the \emph{sensitivity} $\text{sen}(i)$ of a check vertex $i \in \sigma$ be the largest value $l$ such that $i$ is adjacent to an element of $B_l$.
Let $A_l$ be the set of vertices in $\sigma$ with sensitivity $l$ and $a_l = \card{A_l}$. Lastly, define $q_l$ recursively: 
$$q_c = 0, \qquad q_{l} = (q_{l+1} + a_{l+1}) \bmod (l+1) \text{ for } l={1,2, \dots c-1}. $$
The syndrome height $h(\sigma)\geq h_2(\sigma)$, where:
\begin{eqnarray}
\label{eq:tighter-bound}
h_2(\sigma) = 
\ifrac{a_c}{c} + \ifrac{q_{c-1} + a_{c-1}}{c-1} + \ifrac{q_{c-2}+ a_{c-2}}{c-2} + \dots + \ifrac{q_1 + a_1}{1} = \sum_{l=1}^c \ifrac{q_l + a_l}{l}.
\end{eqnarray}
\end{lemma}

\begin{proof} 
To derive this bound, note any valid correction must include at least one fault next to each check vertex in the syndrome. 
We seek the minimum number of fault vertices required to satisfy this condition for all check vertices in the syndrome.

We proceed by considering sets of fault vertices $B_l$ and check vertices $A_l$ with decreasing sensitivity levels $l$. 
Starting from $B_c$ and $A_c$, each fault vertex in $B_c$ can explain (i.e., be adjacent to) at most $c$ check vertices in $A_c$. 
Thus, selecting $\ifrac{a_c}{c}$ fault vertices from $B_c$ would be sufficient to neutralize (in the best-case scenario) all check vertices in $A_c$, except for a remainder of $a_c \bmod c$ checks that are left un-neutralized.

Next, we move to $B_{c-1}$ and $A_{c-1}$. 
Each fault vertex in $B_{c-1}$ can explain at most $c-1$ check vertices, but only those in $A_{c-1}$ and any remaining un-neutralized checks from $A_c$. 
Therefore, the total number of checks to be neutralized at this stage is $q_{c-1} + a_{c-1}$, where $q_{c-1} = a_c \bmod c$. 
Selecting $\ifrac{q_{c-1} + a_{c-1}}{c-1}$ fault vertices from $B_{c-1}$ ensures that all remaining checks in $A_c$ and $A_{c-1}$ are neutralized, except for a new remainder $q_{c-2} = (q_{c-1} + a_{c-1}) \bmod (c-1)$.

This process repeats recursively. 
At each step $l$, the remainder $q_l = (q_{l+1} + a_{l+1}) \bmod (l+1)$ represents the number of checks left un-neutralized after accounting for faults in the previous step. 
To neutralize the checks at step $l$, we select $\ifrac{q_l + a_l}{l}$ fault vertices from $B_l$.
Summing over all steps from $l=1$ to $c$, provides the given lower bound on the syndrome height $h(\sigma)$.
\end{proof}

\begin{algorithm}
    \caption{Height-bound computation}
    \label{alg:height-bound-comp}
    \begin{algorithmic}[1] 
        \Procedure{HeightBound}{$\sigma$}
            \State $a_l \gets 0 \quad \forall l=1, \dots c$ \Comment{Initialize $a_l$ vector}
            \For{$i \in \sigma$} \Comment{Populate $a_l$ vector}
                \State $\text{sen}(i) = 1$
                \For{$j \in \calN(i)$} \Comment{Find sensitivity $\text{sen}(i)$ for all checks}
                    \State $\text{sen}(i) \gets \max(\text{sen}(i), \card{\calN(j) \cap \sigma})$
                \EndFor
                \State $a_{\text{sen}(i)} \gets a_{\text{sen}(i)} + 1$
            \EndFor
            \State $q \gets 0$ \Comment{Initialize remainder}
            \State $h_2 \gets 0$ \Comment{Initialize bound}
            \State $l \gets c$
            \While{$l > 0$} \Comment{Find bound}
                \State $h_2 \gets h_2 \ifrac{q + a_l}{l}$ \Comment{Update bound}
                \State $q \gets (q+ a_l) \bmod l$ \Comment{Update remainder}
                \State $l \gets l-1$
            \EndWhile
            \State \textbf{return} $h_2$ 
        \EndProcedure
    \end{algorithmic}
\end{algorithm}

We provide \alg{height-bound-comp} that computes the syndrome-height lower bound in \eq{tighter-bound}.
Let us briefly consider the time complexity of this computation for any syndrome $\sigma$ for a code code family.
Let constants $r$ and $c$ be the max row and column weights of any matrix $H$ in the code family.
Finding the number of syndrome vertices a fault vertex touches is $O(c)$, while doing this for all faults around a specific check is then $O(r c)$. 
Since this is repeated for all checks in the syndrome, we end up with a time complexity of $O(\card{\sigma} r c)$ for the calculation of all of the $a_l$s. 
The calculation of the bound from that vector is just an additional $O(c)$, such that the total time complexity of \alg{height-bound-comp} is $O(\card{\sigma} r c)$.
For better time complexity, one could use a different algorithm that reuses partial height-bound calculations from parents in the decision tree for child nodes, since their syndromes are only slightly updated.

\textbf{Bounds from syndrome subsets: }
The bound in \eq{tighter-bound} still holds if we evaluate it for a subset of the syndrome $\sigma' \subset \sigma$, such that $\synht(\sigma) \geq h_2(\sigma')$.
It holds because the bound in \eq{tighter-bound} arises from the requirement that there must be at least one fault in the neighborhood of every syndrome vertex, and removing vertices from the syndrome cannot increase the number of faults required to satisfy this constraint.
One important case arises when $G$ has $k$-colorable check vertices (which is the case for the color codes and bivariate-bicycle codes as seen in the next section).
Then taking $\sigma'$ as the syndrome checks of one color, the fault degree is at most 1 such that $c_1 = \card{\sigma'}$, and \eq{tighter-bound} reduces to $h(\sigma) \geq \synht_2(\sigma') = |\sigma'|$. 
See \fig{bound-example}, where we provide an example on the 2D color code.

Another important case is where $\sigma$ separates into disjoint clusters such as $\sigma = \sigma_A  \sqcup \sigma_B$ where $\calN(\sigma_A)$ and $\calN(\sigma_B)$ are disjoint.
In this case, the bounds we obtain for each cluster can be added together to form a bound on $h(\sigma) \geq h_2(\sigma_A)+h_2(\sigma_B)$.
This generalizes straightforwardly to more than two clusters. 

A final general comment is that in some cases various height lower bounds may apply, in which case we can always use whichever happens to be the tightest for a given syndrome.

\begin{figure}[ht]
  \centering
    (a)\includegraphics[width=0.25\textwidth]{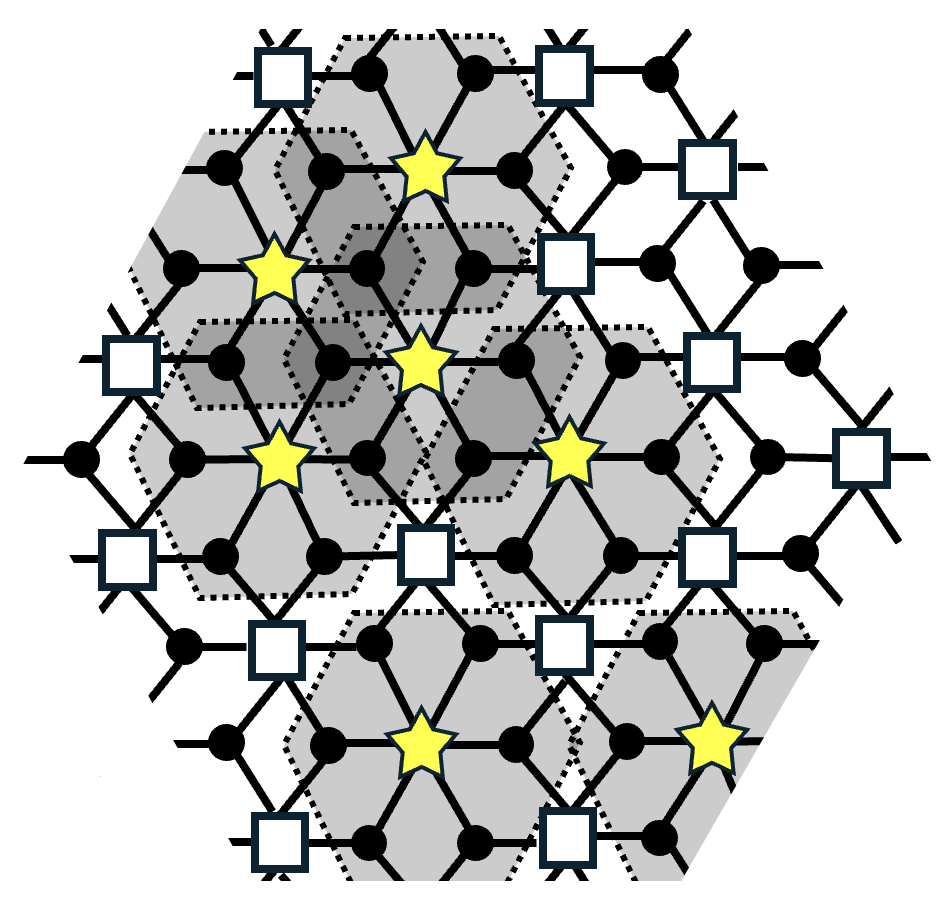}
    (b)\includegraphics[width=0.25\textwidth]{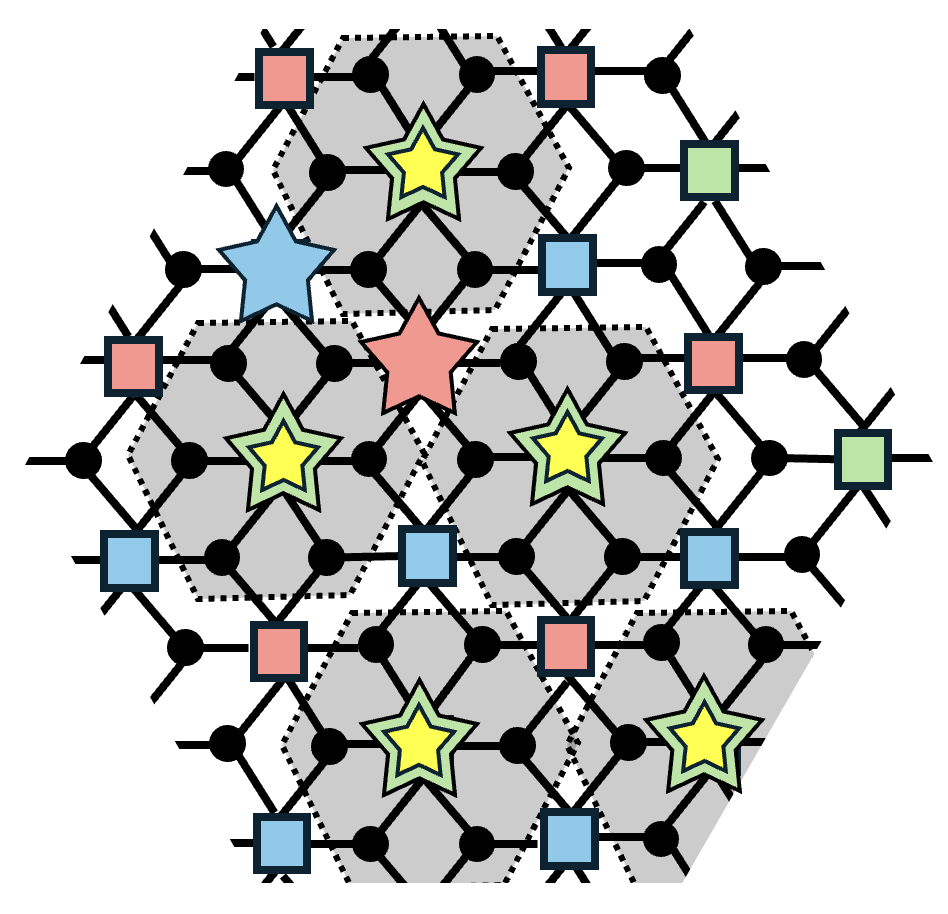}
     \caption{
     Our syndrome height bounds are based on the fact that there must be at least one fault in the neighborhood (gray) of each syndrome node.
     (a) The naive bound $h_1(\sigma) = \lceil 7/3 \rceil =3$ does not take structure of the syndrome neighborhoods into account, whereas the more refined bound $h_2(\sigma) = 4$ does and can be tighter. 
     (b) Using the green-colored syndrome subset $\sigma'$, provides an even tighter bound  in this case $h_2(\sigma') = |\sigma'| = 5$.
     }
  \label{fig:bound-example}
\end{figure}

We provide some other height-lower bounds and prove relations between them in \app{bounds-relations}.

\subsection{Numerical results for color codes and bivariate bicycle codes with perfect measurements}
\label{sec:min-weight-numerics}

Here, we numerically evaluate the runtime of height-bound DTD (defined in \sec{cost-refinement}), measured by the number of explored decision tree nodes.
We use two CSS code families: the standard triangular color codes~\cite{bombin2006topological,Bombin2007} (see \fig{example-codes}(a) in \sec{decoding-scenarios}) and a set of bivariate bicycle codes\footnote{The bivariate bicycle codes are taken from Table 3 in~\cite{bravyi2023higharxiv}, which appears in the arXiv version but not the journal version of the paper.} presented in Ref.~\cite{bravyi2023higharxiv}, including the gross code (see \fig{example-codes}(b) in \sec{decoding-scenarios}). 

We assume independent $X$ and $Z$ errors on each qubit, with perfect stabilizer measurements (see \sec{noise}), and decode $X$ and $Z$ errors separately. 
The color codes and the bivariate bicycle codes we study all have 3-colorable $X$-type ($Z$-type) check vertices\footnote{The 3-colorable property is well known for color codes~\cite{bombin2006topological} and we learned it held for bivariate bicycle codes through private correspondence with Ted Yoder and Sergey Bravyi~\cite{yoder2024private}.} (see \fig{example-decoding-graphs}), meaning that each check vertex can be assigned a color red r, green g, or blue b, such that no two check vertices of the same color are connected to the same fault vertex in the decoding graph $G_X$ ($G_Z$).
This allows us to apply the following lower bound:  
\[
h(\sigma) \geq \max(h_2(\sigma), |\sigma^{\text{r}}|, |\sigma^{\text{g}}|, |\sigma^{\text{b}}|),
\]
for $h_2(\sigma)$ in \eq{tighter-bound}, and where $\sigma = \sigma^{\text{r}} \sqcup \sigma^{\text{g}} \sqcup \sigma^{\text{b}}$ is the syndrome decomposed into color components.

\begin{figure}[ht]
  \centering
  (a)\includegraphics[width=0.4\textwidth]{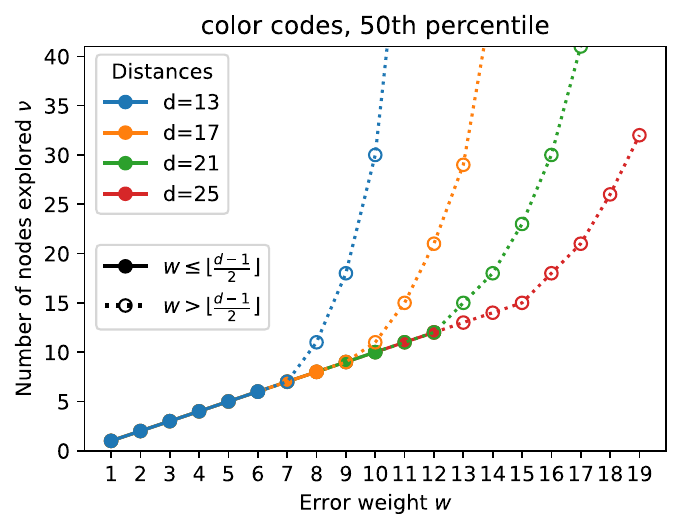}
  (b)\includegraphics[width=0.4\textwidth]{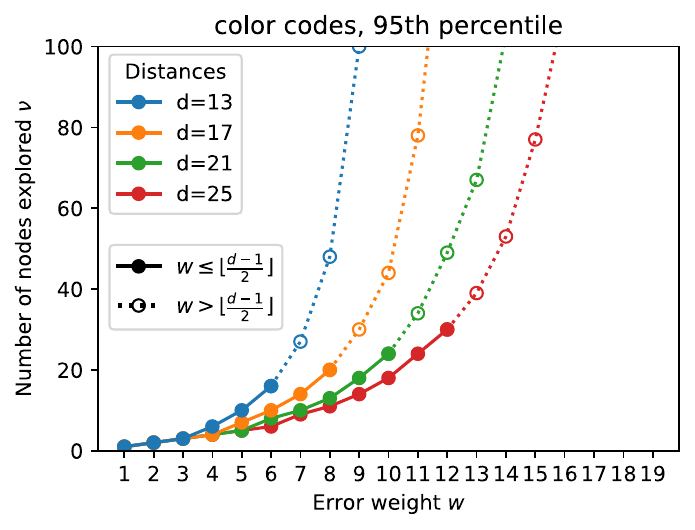}
  (c)\includegraphics[width=0.4\textwidth]{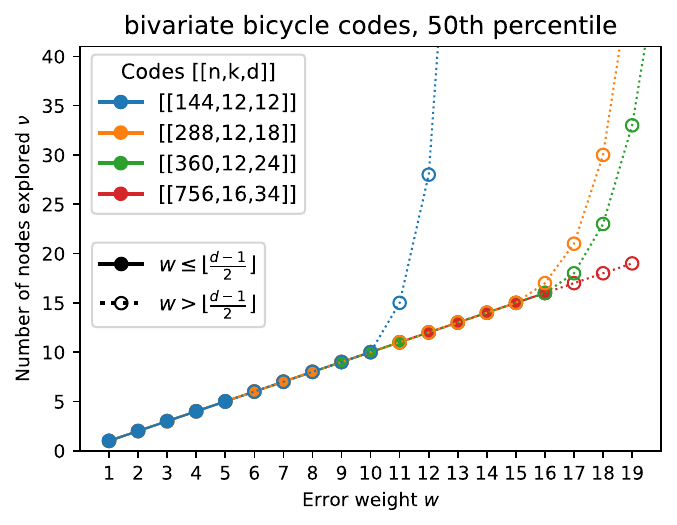}
  (d)\includegraphics[width=0.4\textwidth]{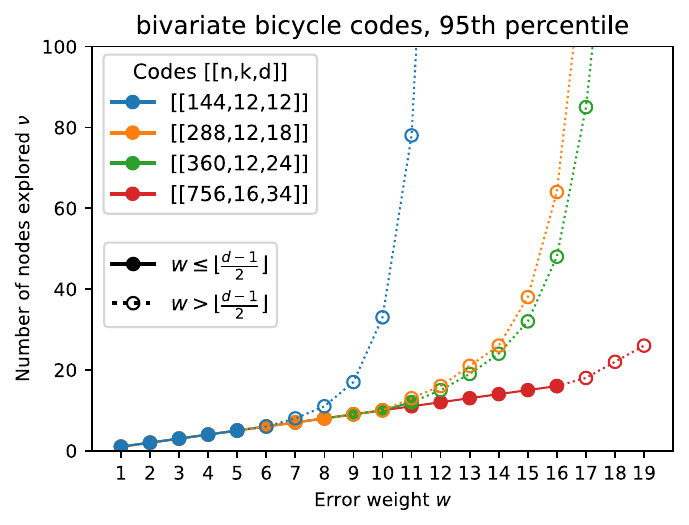}
     \caption{
    Number of explored decision tree nodes $\nu$ with height-bound DTD.  
    (a), (b): Median and 95th percentile of $\nu$ for color codes.  
    (c), (d): Median and 95th percentile of $\nu$ for bivariate bicycle codes.  
    For all codes studied, the median $\nu$ is minimal ($\nu = w$) for errors of weight $w < d/2$, but increases for $w > d/2$.  
    For three of the four bivariate bicycle codes studied, the 95th percentile $\nu$ is also minimal ($\nu = w$) for $w < d/2$.  
    (All error bars are smaller than the markers).
     }
   \label{fig:minweight-results}
\end{figure}

To evaluate the runtime of height-bound DTD for each code, we use it to correct uniformly sampled weight-$w$ $X$-errors and record the number of explored decision tree nodes, reporting the median and 95th percentile in \fig{minweight-results}. 
The results are very promising: for all tested codes, the median was optimal (equal to $w$) for all $w < d/2$, where $d$ is the code distance (or its upper bound if unknown). 
Restricting to the bivariate bicycle codes, the results are even better: for most cases with $w < d/2$, the 95th percentile was also optimal (equal to $w$). 
By contrast, in \fig{minweight-main-results} in \sec{intro} we see that using brute-force breadth-first exploration the median number of explored nodes grows exponentially with $w$.

In \fig{fitting_alpha}, we examine the tails of the runtime distribution for the color code family. 
For odd distances $d$, we uniformly sample weight-$w = \lfloor d/2 \rfloor$ errors (the largest correctable weight) and plot the number of explored nodes for various percentiles in \fig{fitting_alpha}(a). 
We observe that the data for each percentile fits a polynomial $\nu =  w^\alpha$ reasonably well. 
In \fig{fitting_alpha}(b), we plot the fitted $\alpha$ values, finding $\alpha = -0.132\log_2 \epsilon + 0.843$ for the percentile $100 \cdot (1 - \epsilon)$.   
For fixed (finite) $\epsilon$, this indicates polynomial scaling of $\nu$ with $w$, but as $\epsilon \to 0$, $\alpha$ diverges, implying super-polynomial worst-case runtime with $w$.

\begin{figure}[H]
  \centering
    (a)\includegraphics[width=0.45\textwidth]{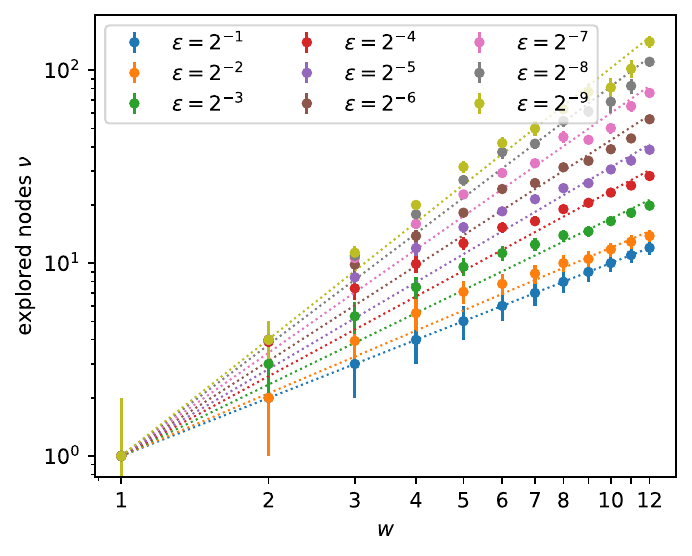}
    (b)\includegraphics[width=0.45\textwidth]{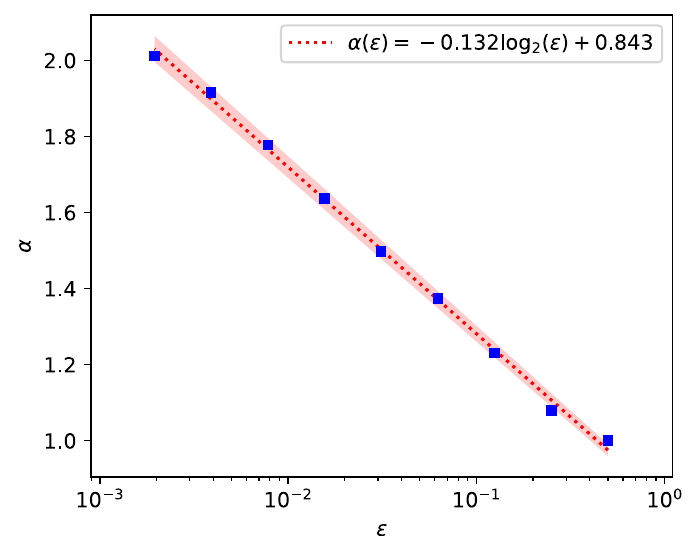}
     \caption{
    Scaling behavior of decoder runtime for the color code family. 
    (a) Number of explored nodes ($\nu$) for weight-$w=\lfloor d/2 \rfloor$ $X$-errors, plotted across percentiles $100\cdot(1-\epsilon)$. 
    For each $\epsilon$, we fit the function $\nu = w^{\alpha}$, extracting fit parameter $\alpha(\epsilon)$.
    (b) A plot of $\alpha(\epsilon)$ versus $\epsilon$ (log scale).
   }
  \label{fig:fitting_alpha}
\end{figure}

\textbf{Numerical details: }
In what follows, we provide further details needed to reproduce the data in this section.
For the BP subroutine of \sub{minweight-enhanced-cost} we use $t_\text{end}=12$ rounds. 
In the case of 2D color codes, there is a symmetry between $X$ and $Z$ such that $X$ errors and $Z$ errors behave identically, so we restrict to $X$ errors.
For the BB codes we present only the data for $X$ errors as we observed almost identical performance for $Z$ errors.
To estimate uncertainties for the median and 95th percentile values in \fig{minweight-results}, we resampled the data but observed values of $\nu$ were all within $\pm 1$, which is smaller than the marker size. 
In \fig{fitting_alpha}(a), we use whichever is larger between $\pm 1$ and the resampling estimate as the uncertainty for each data point. 
In \fig{fitting_alpha}(b), we plot the $\chi^2$-fit estimate of $\alpha$ along with its corresponding uncertainty. 
The uncertainties for the data points here are below the marker size. 

\textbf{Comparison with MaxSAT: }
In \app{maxsat-comparison}, we compare the runtime of our height-bound DTD to MaxSAT decoding data from Ref.~\cite{noormandipour2024maxsatdecodersarbitrarycss}. 
Encouragingly, our algorithm appears significantly faster in the low-error-rate regime, though the MaxSAT decoder appears to outperform it at error rates closer to the threshold.
While this comparison provides qualitative insights, it should be interpreted cautiously since the runtimes were measured on different machines, with implementations that were not optimized for absolute runtime.
A key conceptual difference between these algorithmic approaches is that incorporating code-specific structure (e.g., tighter height bounds) appears clearer for height-bound DTDs than for MaxSAT decoders.

\clearpage

\section{Belief-propagation decision-tree decoder}
\label{sec:fast-cost}
This section introduces a heuristic decision-tree decoder which may be useful for fast decoding. 
The decoder uses belief propagation to assign costs to nodes resulting in an approximate depth-first exploration. 
We provide and explain the exploration subroutine in \sec{bp-exploration} and compare the decoder’s speed and accuracy with a standard BP-OSD algorithm in \sec{circuit-noise-numerics}.

\subsection{Exploration using belief propagation}
\label{sec:bp-exploration}
Our heuristic algorithm calculates a cost update using belief propagation (BP). 
Specifically, for a decision tree node with fault set $F$, syndrome $\sigma$ and cost $C$, we apply decimated BP (see \sec{decoding-algorithms}), which outputs the LPR $\bar{\Lambda}_j =\bar{\Lambda}_\text{BP}(j \,| \, \sigma, F) $ for each fault $j$. 
The cost update $\Delta C$ for each child is then computed, resulting in a new child cost of $C + \Delta C$, where $\Delta C$ is given by:
\begin{equation}
\label{eq:heuristic-cost}
\Delta C(\bar{\Lambda}_j) = \frac{13}{\pi} \arctan\biggl( \frac{\bar{\Lambda}_j}{2}-1 \biggr) + \frac{11}{2}.
\end{equation}
This monotonically increasing function of $\bar{\Lambda}_j$ helps to stabilize the decoder  by keeping the cost within a finite range, even though $\bar{\Lambda}_j$ is unbounded (with $\bar{\Lambda}_j \rightarrow -\infty$ ($+\infty$) indicating BP is certain that $j$ is (is not) part of a minimum-weight correction).  
The constants in \eq{heuristic-cost} are somewhat arbitrary; the selected values give $\Delta C(-\infty) = -1$ and $\Delta C(+\infty) = 12$, reflecting changes in the decimated syndrome height $h_F(\sigma)$. 
Specifically, $h_F(\sigma)$ decreases by 1 when a fault is in a minimum-weight correction and increases by $s-1$ otherwise, where $s$ is the weight of the smallest stabilizer containing the fault.

From the heuristic cost update in \eq{heuristic-cost}, we define the BP exploration \sub{bp-cost}.
When \texttt{Explore} \sub{bp-cost} is used in \alg{decision-tree-decoder}, we call the resulting decoder \emph{BP-DTD}.

\begin{subroutine}[ht]
    \caption{BP-guided exploration}
    \label{alg:bp-cost}
    \begin{algorithmic}[1] 
        \State $\set{\bar{\Lambda}_j}, \, \hat{F}_\text{BP} \gets \texttt{BPDecimated}(s,F)$ \Comment{Run BP on $\sigma$ with decimation $F$}
        \If{$\sigma = H\hat{F}_\text{BP}$} \Comment{If BP converged, invoke early-exit condition}
            \State \textbf{return} $F \cup \Hat{F}_\text{BP}$
        \EndIf
        \State \textbf{pick} $i \in \sigma$
        \For{$j  \in  \calN(i) \! \setminus \! F = \{j_1, j_2, \dots j_b \}$}
        \State $\Delta C = \frac{13}{\pi} \arctan( \frac{1}{2}(\bar{\Lambda}_j-2)) + \frac{11}{2}$
            \State $C^{(j)} \gets C + \Delta C $
        \EndFor
    \end{algorithmic}
\end{subroutine}

Note that in \sub{bp-cost} if BP converges to a valid partial correction $\sigma = H\Hat{F}_\text{BP}$, the algorithm exits early and the decision tree decoder returns the full correction $F \cup \Hat{F}_\text{BP}$.

\subsection{Numerical results for the gross code with circuit noise}
\label{sec:circuit-noise-numerics}

Here we numerically evaluate the runtime and accuracy of the belief-propagation decision tree decoder (BP-DTD) using BP exploration~\sub{bp-cost}. 
The decoder is tested on the gross code, a $[[144, 12, 12]]$ bivariate-bicycle code from \cite{bravyi2024high} (see \fig{example-codes}(b) in \sec{decoding-scenarios} for details) assuming circuit noise.

We compare the accuracy and runtime of the BP-DTD and BP-OSD algorithms by generating random samples at physical error rates at physical error rates between $p=0.001$ and $p=0.004$, recording each decoder’s runtime and whether it succeeds for each sample.  
Since no system-independent metric exists, we report absolute runtimes measured on the same CPU cluster, noting that implementation details and processor type may influence the results.  
\fig{logical_error_rates} shows that BP-DTD outperforms BP-OSD in the whole probed error rate regime. 
\fig{decoding-time-dist3} presents runtime distributions, with logarithmic binning and relative frequencies shown as the fraction of samples per bin.  
BP-OSD shows a bimodal distribution, with the later mode corresponding to cases requiring the OSD stage. 
In contrast, BP-DTD has a smoother distribution, lower average runtime at low error rates, but longer tails at higher error rates.

\begin{figure}[H]
  \centering
    \includegraphics[width=0.9\textwidth]{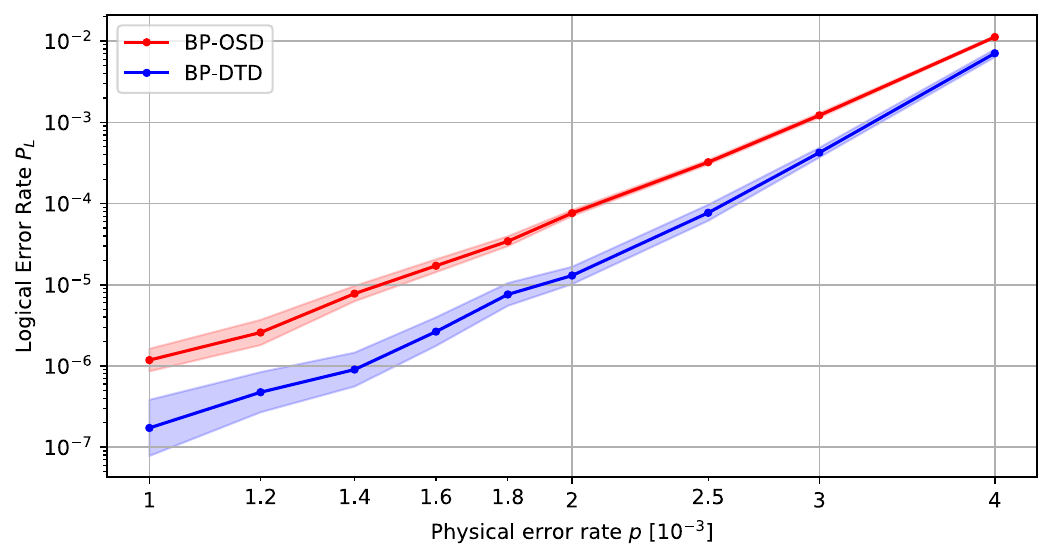}
     \caption{
     Comparing the logical error rate for the gross code for a range of circuit-noise error rates when decoding with BP-DTD versus BP-OSD (with combination-sweep 10). 
     Note here that we do not divide by the number of error correction cycles.
     }
  \label{fig:logical_error_rates}
\end{figure}

\begin{figure}[ht]
  \centering
    \includegraphics[width=0.8\textwidth]{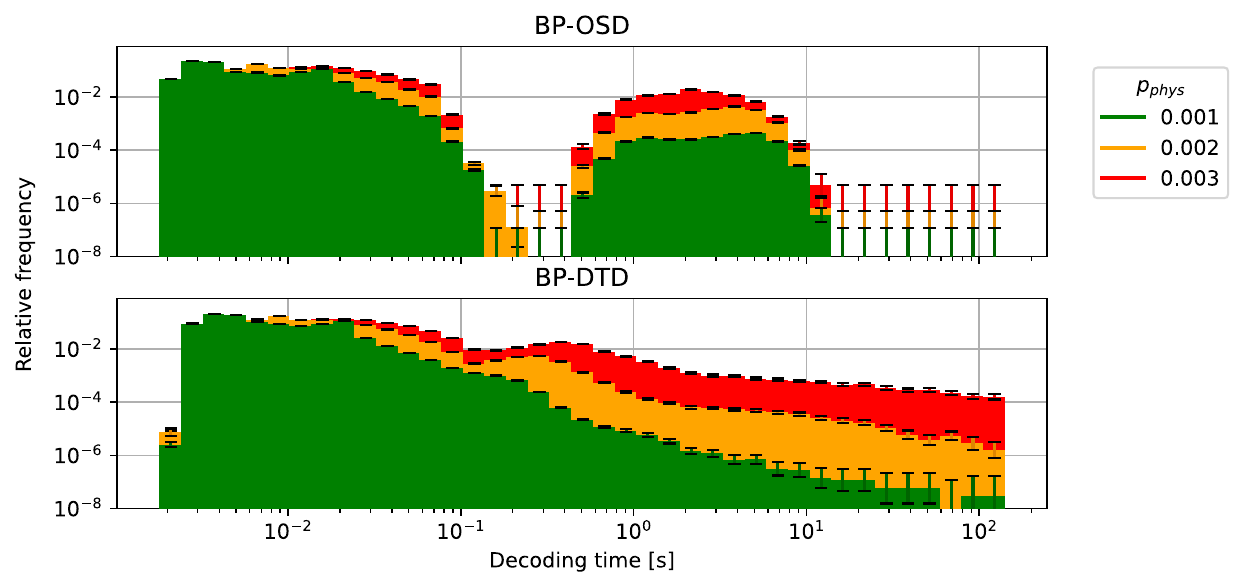}
     \caption{
    Runtime distributions of BP-DTD and BP-OSD for different circuit-noise error rates. 
    BP-OSD shows a bimodal distribution, where the later mode corresponds to cases requiring the OSD stage. 
    In contrast, BP-DTD has a smoother distribution with lower average runtime. 
    (Note that the data is separated into time-bins with width that grows logarithmically with decoding time.)
     }
  \label{fig:decoding-time-dist3}
\end{figure}

In many scenarios, both the decoding accuracy and runtime are crucial, requiring trade-offs~\cite{delfosse2023choose}. 
We propose using \emph{cutoff-time performance curves}, which plot the logical failure probability against the cutoff time $T$. 
Logical failure occurs either by exceeding $T$ or by decoding within $T$ but outputting an incorrect correction. 
In \fig{ler-vs-tcutoff-intro} (\sec{intro}), we compare the cutoff-time performance curves for BP-DTD and BP-OSD at $p=0.001$ error rate, highlighting BP-DTD’s significantly lower failure rates at intermediate cutoff times.

\textbf{Numerical details:}  
In what follows, we provide further details needed to reproduce the data in this section. 
The circuit parity-check matrix for the gross code is obtained from the publicly available code in Ref.~\cite{gong2024}, which uses the same syndrome extraction circuit as \cite{bravyi2024high} and the standard linearized circuit noise model (see \sec{noise}) with $R = 12$ syndrome measurement rounds plus one perfect round.  
Here we decode only the $X$-type errors. 
For BP-OSD, we use combination sweep setting 10 and $t_\text{end}=100$ BP rounds as a pre-decoder (see \app{bp-osd}), with Roffe’s standard Cython implementation \cite{roffe2020, Roffe_LDPC_Python_tools_2022}. 
We found that varying the combination sweep order had negligible impact on runtime distributions, so results are shown only for BP-OSD-10.  
For BP-DTD, parameters are $t_\text{end}=100$ at the root node, $t_\text{end}=12$ for other nodes, buffer length $l_\text{buff}=8$, and a node cap of 50000, after which decoding is declared a failure.  
The DTD code, written in Python, calls a modified Cython-based BP subroutine adapted from Roffe’s version to enable decimation and buffering.

\clearpage

\section{Outlook and future directions}
\label{sec:conclusion}

We conclude by first highlighting two potential applications of the height-bound DTD.
One promising application is as a pre-decoder: if it converges quickly, it guarantees a minimum-weight correction; otherwise, a fast fallback decoder can provide a correction without such guarantees.
Another application is to provably determine the distances of specific qLDPC codes by adapting a method described in Ref.~\cite{bravyi2024high} for finding distance upper bounds. 
By replacing the heuristic BP-OSD with the provable height-bound DTD, we can identify the exact distance rather than merely upper-bounding it (see \app{finding-code-distances}).

 Finally, we identify several promising directions for future research:

\begin{itemize}

    \item \textbf{Leveraging stabilizer equivalence: }
    In the height-bound decoder, the algorithm terminates when it finds a correction and ensures no lower-weight correction exists. 
    However, two corrections are equivalent if they differ by a stabilizer.
    A faster version of the height-bound decoder (with equally strong correction guarantees) could more aggressively search for a correction and then ensure no lower-weight \emph{non-equivalent} correction exists (i.e., corrections differing by a logical operator). 
    Upon finding a valid correction, the decoder could aggressively prune the remaining search space by applying bounds that tighten the bound for each node given the known correction.\footnote{Specifically, for a DT node $F$ and known correction $\hat{F}$, any non-equivalent correction $F'$ descending from $F$ must satisfy $|F'| \geq 2|\hat{F} \cap F| + d - |F|$. 
    The closer $F$ is to $\hat{F}$, the tighter this bound becomes.}

    \item \textbf{Alternative versions of height-bound DTD: }
    Exploring alternative tie-breaking methods beyond the decimated BP used in this work may further improve decoding speed. 
    Also, while the current height-bound DTD applies to any qLDPC decoding problem, including circuit-level noise, it does not exploit non-uniform fault probabilities. 
    Developing lower bounds on syndrome height for non-uniform fault settings could enhance its effectiveness in this context. 
    It could also be interesting to consider variants of the decoder which provide faster solutions but with weaker guarantees (perhaps outputting a correction along with a bound on how far it is from having minimum weight).
 
    \item \textbf{Other provable qLDPC decoders: } 
    Ensemble decoding suggests it would be very useful to identify other provable decoders (either provable performance or provable runtime~\cite{demarti2024almost} but not both)?   
    Another question is: What kind of syndrome height bounds would allow us to prove something non-trivial about the worst-case run time for an asymptotic code family? 

    \item \textbf{Other applications of DTD techniques: }
    It may be possible to adapt decision tree decoding techniques to solve other problems, for example to identify provable weight gaps between the minimum-weight correction and any logically non-equivalent correction (to bound the failure rate of specific fault-tolerant protocols).
    
    \item \textbf{Improving heuristic decoders: }  
    Further optimization of BP-DTD’s cost function could yield significant performance improvements, with optimal parameters likely depending on the specific QEC code and error rate regime. 
    One promising direction may be to train a neural network to estimate the syndrome height, which could then be used as a cost metric for the decision tree decoder. 
    Since the model only needs to output a single number rather than a complete correction, training may be simpler and more efficient compared to other neural-network-based decoders.

    \item \textbf{Cutoff-time performance curves: }
    Comparing existing and new heuristic decoders could help identify which methods perform best under different time constraints and error regimes.

\end{itemize}

In summary, improving cost functions, leveraging stabilizer equivalence, and exploring ensemble decoding strategies are promising directions for both provable and heuristic decision tree decoders. 
A primary challenge remains developing decoders that apply to any qLDPC code, return minimum-weight corrections and have provably efficient worst-case runtime (or proving that such a decoder cannot exist).

\textbf{Acknowledgments: }
This work benefited heavily from guidance and insightful discussions with thesis supervisors James R. Wootton and Joseph M. Renes, and also feedback from Nicolas Delfosse, Anqi Gong, Tomas Jochym-O'Connor, and Anirudh Krishna.
We thank Ted Yoder and Sergey Bravyi for pointing out that the set of bivariate bicycle codes we test our height-bound decoder on have 3-colorable decoding graphs.

\clearpage
\bibliographystyle{alpha}
\bibliography{references}

\clearpage
\appendix

\section{Appendices}

\subsection{Noise models and testing decoders numerically}
\label{sec:noise}

Here we define the experiments and noise models that we use for the studies in \sec{min-weight-numerics} and \sec{circuit-noise-numerics}, and specify how we test the decoders numerically.

We perform \emph{Quantum Memory} experiments using \emph{CSS} codes, that is, we are in the setting of logical memory circuits for CSS-type stabilizer codes (for which we refer back to \tab{decoding-matrices-summary} and \sec{stab-codes}, \sec{space-time-codes}). Here, decoding is performed independently for $X$-type noise via the $X$-type check matrix $H_X$ (or $H_Z$ for $Z$-type noise). $H_X$ and $H_Z$ (and equivalently the corresponding decoding graphs $G_X$ and $G_Z$) depend on the exact noise model used and the number of stabilizer extraction rounds chosen, which we discuss below.

\textbf{Data qubit noise}: Each data qubit experiences an $X$ error with probability $p$, and (independently) a $Z$ error with probability $p$, and then the stabilizer generators are measured perfectly. Because the measurements are perfect, performing a single round $R=1$ is sufficient. In this setting the check matrices $H_X$, $H_Z$ reduce to the check matrices $H_X^\calS$, $H_Z^\calS$ of the underlying CSS-type stabilizer code, and the decoding graphs $G_X$, $G_Z$ are precisely its tanner graphs $\calT_X$, $\calT_Z$. We use this noise model in \sec{min-weight-numerics} to study the height-bound DTD on the color code family (shown in \fig{example-decoding-graphs}(a)) and the bi-variate bicycle code family (a member of which is shown in \fig{example-decoding-graphs}(b)).

\textbf{Circuit depolarizing noise}:
An explicit detailed circuit built from one- and two-qubit operations (single qubit preparations and measurements in the Pauli basis, and Clifford unitaries) which measures the stabilizers for $R-1$ rounds, and then another round of perfect measurements is included. 
Each operation fails with probability $p$.
When a preparation fails, an orthogonal state is prepared.
When a measurement fails, the outcome is flipped.
When a Clifford unitary fails, a random non-trivial Pauli operator is applied to its support.\footnote{Note that as we have described this model, the different failure modes (for example each of the 15 non-trivial Paulis which can occur when a CNOT gate fails) are mutually exclusive, whereas we have assumed all faults occur independently.
In Ref.~\cite{chao2020} it was shown that there is an exact map from this exclusive noise model to an independent one where each of the aforementioned failure modes occur with probability $p' = p/15 +O(p^2)$. 
Here, as in other works in the literature~\cite{bravyi2024high}, we drop the non-linear corrections for simplicity.}
We use this noise model in \sec{circuit-noise-numerics} to study the decoding of the gross code which was shown in \fig{example-decoding-graphs}(b).
To do so, we take $R=12+1$, i.e. 12 rounds of noisy stabilizer measurements (using the stabilizer extraction circuit used in \cite{bravyi2024high}), followed by a single round of perfect measurements.

\textbf{Testing decoders: }
We can test a decoder by sampling among the possible fault configurations and for each fault configuration $F$, running the decoder on the syndrome $\sigma = H F$, and then checking whether the decoder output $\hat{F}$ succeeds or fails. 
(Since we are exclusively using CSS codes and always only decoding either $X$-type or $Z$-type errors, we drop the subscripts $X$ and $Z$.
We consider two kinds of samples: (1) sampling uniformly from the set of all fault configurations of a given (integer) weight $w$, and (2) sampling according to the probability distribution $\Prob(F)$ defined in \defn{decoder}.

If $N_\text{trials}$ are sampled, of which $N_\text{fail}$ fail and the rest succeed, we estimate the logical failure probability:
$$ P_L = \frac{N_\text{fail}}{N_\text{trials}}.$$

We report the uncertainty in the estimate of $P_L$ by the Wilson Score Interval (for a standard normal interval half-width of 2):
$$ P_L \in \frac{1}{1 + 4 / N_\text{trials}} \left( \frac{N_\text{fail}}{N_\text{trials}}+ \frac{2}{N_\text{trials}} \pm \frac{2}{N_\text{trials}} \sqrt{N_\text{fail}(1-\frac{N_\text{fail}}{N_\text{trials}}) + 1 }\right) $$

In addition to the logical failure rate, we also study other features of decoders, such as their run time. 
For these, we find it informative to provide more information about the sample distribution than that captured by the mean, and so we often report the median value or more generally the value at a particular percentile among those observed in the sample. 

\subsection{Belief-propagation with ordered statistics decoding}
\label{app:bp-osd}

The advantages of BP are high speed and parallelizability. 
However, because of its inherent locality, BP does not always converge to a solution consistent with the syndrome even for an infinite amount of message-passing rounds \cite{panteleev2021degenerate, gong2024}. 
In quantum error correction, BP is thus often used as a predecoder. 
In the case BP fails to produce a valid solution within a certain maximum number of iterations, its output will be used by a different (post-processing) decoding strategy such as ordered-statistics decoding (OSD) \cite{roffe2020}, matching or union-find (UF) \cite{higgott2023improveddecodingcircuitnoise}. 

In \emph{zero-order OSD} (BP-OSD-0), the degeneracy of the decoding problem is lifted by reducing the search-space such that it allows for a unique solution. For that the checkmatrix is reorderd such that the $\rank (H)$ linearly independent columns with the highest LPR are on the left side (in $S$) followed by the remainder of the columns on the right ($T$), which are also ordered according to LPRs:
\begin{equation*}
    H \Pi_\text{BP} = \begin{bmatrix}
        S & T
    \end{bmatrix}
\end{equation*}
where $\Pi_\text{BP}$ is the corresponding $N \times N$ permutation matrix, $S$ is a $M \times \rank(H)$ matrix and $T$ is a $M \times (N-\rank(H))$ matrix.
There is a unique solution $x$ to $S x = \sigma$ (given by $x = A^+ \sigma$, where $A^+$ is the pseudo-inverse given $\sigma$ always lies in the column-span of $S$), which defines the solution to the original problem by $\Hat{F} = \Pi_\text{BP} \begin{bmatrix}
    x \\ 0
\end{bmatrix} $:
\begin{equation*}
    H \Hat{F} = H \Pi_\text{BP} \begin{bmatrix}
    x \\ 0
\end{bmatrix} = \begin{bmatrix}
        S & T
    \end{bmatrix} \begin{bmatrix}
    x \\ 0
\end{bmatrix} = S x = \sigma
\end{equation*}

In \emph{higher-order OSD}, we additionally allow faults in the support of $T$. That is, we allow vectors $t$ and define $\Hat{F}_t = \Pi_\text{BP} \begin{bmatrix}
    x + S^+ Tt \\ t
\end{bmatrix} $, which is always a solution as 
\begin{equation*}
    H \Hat{F}_t = H \Pi_\text{BP} \begin{bmatrix}
    x + S^+ Tt \\ t
\end{bmatrix} = \begin{bmatrix}
        S & T
    \end{bmatrix} \begin{bmatrix}
    x + S^+ Tt \\ t
\end{bmatrix} = S x + Tt + Tt = \sigma
\end{equation*}
Again, the use of the pseudo-inverse $S^+$ is well justified by the fact that $Tt$ is always in the column span of $S$.

The chosen solution is then the one with the minimal weight $\Hat{F} = \mathop{\arg \max}_{t} w(\Hat{F}_t)$.
In principle, there are $2^{N-\rank(H)}$ choices for $t$, so this search quickly becomes untractable.
In practice, there are two main approaches:
\begin{itemize}
    \item The \emph{exhaustive-search order $\lambda$} (BP-OSD-ES-$\lambda$), which searches over all $t$ with support on the $\lambda$ left-most columns in $T$, giving $2^\lambda$ configurations.
    \item The \emph{combination-sweep order $\lambda$} (BP-OSD-CS-$\lambda$), which searches through all configurations of $t$ which have hamming weight $1$ and all hamming-weight $2$ configurations with support on the first $\lambda$ columns of $T$. In total this yields $N- \rank(H) +\binom{\lambda}{2}$ configurations.
\end{itemize}

For fixed number of configurations checked, combination-sweep order performs slightly better than exhaustive-search order \cite{roffe2020}. Throughout this work, we implicitly use combination-sweep order and write BP-OSD-$\lambda$ as short for BP-OSD-CS-$\lambda$.

BP-OSD is implemented in the python package \emph{ldpc} by \cite{Roffe_LDPC_Python_tools_2022}

\subsection{Relations between different syndrome-height bounds}
\label{app:bounds-relations}
Here we consider bounds on the syndrome height from \sec{provable-cost} and provide a few others.
While the other bounds are less tight (which we show below), they may be preferred for their simplicity for some applications.
We try to make this discussion self-contained, so first we review the setting.
Consider a bipartite graph $G$ with two types of vertices: check vertices and fault vertices, such that no two check vertices share an edge and no two fault vertices share an edge.
Let $c$ be the maximum number of check vertices touching any individual fault vertex.
Let the error $F$ be an (unknown) set of fault vertices.
The syndrome $\sigma$ is then the set of check vertices which neighbor an odd number of fault vertices in $F$.

Let $B_l$ be the set of fault vertices in $G$ (which may or may not be in $F$) which touch $l$ vertices in $\sigma$ for $l=1,2,\dots, c$.
Let $b_l$ be the size of the set $B_l$ for $l=1,2,\dots, c$.
Let $\text{sen}(v)$ of a check vertex $v \in \sigma$ be the largest integer $m$ such that $v$ is adjacent to an element of $B_m$ for $m=1,2,\dots, c$.
Let $a_l$ be the number of vertices in $\sigma$ with $\text{sen}=l$ for $l=1,2,\dots, c$.
Given $\sigma$, let $h(\sigma)$ be the minimum size of a fault set $F$ that results in the syndrome $\sigma$.
Consider the following lower bounds bounds for $h(\sigma)$:
\begin{eqnarray}
h_1(\sigma) &=&  \Bigl\lceil \frac{|\sigma|}{c} \Bigr\rceil. \nonumber \\
h_2(\sigma) &=& \ifrac{a_c}{c} + \ifrac{(a_c \bmod c) + a_{c-1}}{c-1} + \ifrac{((a_c \bmod c) + a_{c-1}) \bmod (c-1) + a_{c-1}}{c-2} + \dots. \nonumber\\ \\
h_3(\sigma) &=&  \Bigl\lceil \sum_{l=1}^c \frac{a_l}{l} \Bigr\rceil. \nonumber\\
h_4(\sigma) &=& \Bigl\lceil \frac{\card{\sigma} - |B_c|}{c-1} \Bigr\rceil. \nonumber
\end{eqnarray}

These bounds are related by: 
\begin{eqnarray}
h_2(\sigma) &\geq& h_3(\sigma), \nonumber \\
h_3(\sigma) &\geq& h_1(\sigma), \nonumber \\
h_3(\sigma) &\geq& h_4(\sigma), \nonumber 
\end{eqnarray}
such that $h_2(\sigma)$ is the tightest bound out of the four.

That $h_3(\sigma) \geq h_1(\sigma)$ is pretty straightforward, so we do not include that proof explicitly.
Next we prove that $h_3(\sigma) \geq h_4(\sigma)$. 
Consider a syndrome $\sigma$ with fault vertex sets $B_{c}, B_{c-1}, \dots, B_{1}$, and partition $|\sigma| = a_{c} + a_{c-1} + \dots + a_{1}$.
We define a new function $h_5(\sigma)$ as follows:
\begin{eqnarray}
h_3(\sigma) = \Bigl\lceil \sum_{l=1}^r \frac{a_l}{l} \Bigr\rceil \geq \Bigl\lceil  \frac{a_c}{c} + \frac{a_{c-1} + a_{c-2} + a_{1}}{c-1} \Bigr\rceil = \Bigl\lceil  \frac{a_c}{c} + \frac{|\sigma|-a_c}{c-1} \Bigr\rceil = h_5(\sigma), \nonumber\\
\end{eqnarray}
and therefore the function $h_5(\sigma)$ forms a weaker lower bound than $h_3(\sigma)$.
Now define a subset $\sigma' \subset \sigma$ obtained by removing one of the check vertices from $\sigma$ which touches an element of $B_{c}$.
This means $|\sigma'| = |\sigma|-1$ and $a'_c = a_c - \Delta$ for $1 \leq \Delta \leq c$, since that check vertex was touching at least one fault vertex that was touching $c$ elements of $\sigma$, and some of those check vertices could have been touching more than one element of $B_c$ (in which case they would still contribute to $a'_c$).
Also note that $|B'_c| \leq |B_c|-1$.
Therefore 
$$h_5(\sigma')= \Bigl\lceil  \frac{a_c-\Delta}{c} + \frac{|\sigma|-1-a_c+\Delta}{c-1} \Bigr\rceil = \Bigl\lceil  \frac{a_c}{c} + \frac{|\sigma|-a_c}{c-1} 
 + \frac{\Delta/c-1}{c-1} 
\Bigr\rceil \leq h_5(\sigma).
$$
We repeat this process iteratively (removing a syndrome vertex touching a fault vertex that touches $c$ syndrome vertices), until we end up with a syndrome $\sigma''$ such that $a''_c=0$ (which implies $B''_{c}$ is empty).
We have that:
\begin{eqnarray}
h_3(\sigma) \geq h_5(\sigma) \geq h_5(\sigma') \geq h_5(\sigma'')= \Bigl\lceil   \frac{|\sigma''|}{c-1} \Bigr\rceil \geq
\Bigl\lceil \frac{|\sigma|-|B_c|}{c-1} \Bigr\rceil = h_4(\sigma), \nonumber\\
\end{eqnarray}
where the last inequality is obtained by noting that it must be possible to obtain $\sigma''$ from $\sigma$ by a sequence of removing at most $|B_{c}|$ check vertices.

Lastly we show that $h_3(\sigma) \leq h_2(\sigma)$ since:
\begin{align*}
    \sum_{l=1}^c \frac{a_l}{l} &= \frac{a_c}{c} + \frac{a_{c-1}}{c-1} + \frac{a_{c-2}}{c-2} + \dots \\
    &= \ifrac{a_c}{c} + \frac{a_c \bmod c}{c} + \frac{a_{c-1}}{c-1} + \frac{a_{c-2}}{c-2} + \dots \\
    &\leq \ifrac{a_c}{c} + \frac{a_c \bmod c}{c-1} + \frac{a_{c-1}}{c-1} + \frac{a_{c-2}}{c-2} + \dots \\
    &= \ifrac{a_c}{c} + \frac{a_c \bmod c + a_{c-1}}{c-1} + \frac{a_{c-2}}{c-2} + \dots \\
    &= \ifrac{a_c}{c} + \ifrac{a_c \bmod c + a_{c-1}}{c-1} + \frac{((a_c \bmod c) + a_{c-1}) \bmod (c-1)}{c-1} + \frac{a_{c-2}}{c-2} + \dots \\
    &\leq \ifrac{a_c}{c} + \ifrac{a_c \bmod c + a_{c-1}}{c-1} + \frac{((a_c \bmod c) + a_{c-1}) \bmod (c-1)}{c-2} + \frac{a_{c-2}}{c-2} + \dots \\
    &= \ifrac{a_c}{c} + \ifrac{a_c \bmod c + a_{c-1}}{c-1} + \frac{((a_c \bmod c) + a_{c-1}) \bmod (c-1) + a_{c-2}}{c-2} + \dots \\
    &\vdots
\end{align*}
The left-hand side is still smaller or equal even when taking the ceiling of it (to get $h_3(\sigma)$, because the right-hand side $h_2(\sigma)$ is an integer). 
So rounding up the LHS is at most equal to $h_3$ and will never overtake $h_3$. 
In other words, the bound $h_2$ implies $h_3$. 
We have confirmed by generating some explicit random syndrome instances that there is a separation between the bounds $h_3$ and $h_2$. 

Lastly, we point out that $h_2(\sigma)$ is the tightest lower bound for $h(\sigma)$ that can be obtained from the $a_l$ vector alone.
That is, for any vector $a_l$, there exists a decoding graph $G$ and a syndrome $\sigma$ with the vector $a_l$ such that $h(\sigma) = h_2(\sigma)$.
Of course, there can be other tighter bounds for $h(\sigma)$ that use information beyond the $a_l$ vector.

\subsection{Finding the code distance using a min-weight decoder}  
\label{app:finding-code-distances}

Here we review how to use a decoder (such as height-bound DTD) which is guaranteed to output a min-weight correction to determine the distance of a specific qLDPC code.
This is essentially the approach described in Ref.~\cite{bravyi2024high}, which used BP-OSD to find upper bounds on code distance, but where we use a min-weight decoder in palce of BP-OSD.

For decoding matrices $H \in \mathbb{F}_2^{M \times N}$ and $A \in \mathbb{F}_2^{K \times N}$, construct an extended check matrix $H^{(i)} \in \mathbb{F}_2^{(M+1) \times N}$ for each $i \in {1, \dots, K}$ by appending the $i$th row of $A$ to $H$.
Set the syndrome $\sigma^{(i)} \in \mathbb{F}_2^{M+1}$ such that $\sigma^{(i)}_j = 1$ only for $j = M + 1$ and use the min-weight decoder to find a min-weight correction $\hat{F}^{(i)} \in \mathbb{F}_2^N$, which represents a min-weight logical operator for $H$ that is non-trivial for the $i$th row of $A$.

Note that the decoder may not be guaranteed to terminate quickly (we have no such guarantee for height-bound DTD for example).
Therefore the decoder may not always find $\hat{F}^{(i)}$ for all $i \in {1, \dots, K}$ within an acceptable period of time.
However, if it does, the code distance is given by $d = \min_{i} |\hat{F}^{(i)}|$, since any non-trivial logical operator must be non-trivial for at least one row of $A$.

A potential challenge for this approach when using height-bound DTD is that the added row in $H^{(i)}$ corresponds to a high-degree vertex (at least $d$) in the decoding graph.
This will result in a large number of children in the DT at the first step of the algorithm, which may be somewhat alleviated by a careful basis choice for $A$ so that the vertex degree is not much more than $d$.

\subsection{Finding all min-weight logical operators using height-bounded decision trees}
\label{app:find-min-weight-logicals}

Here we provide an approach to compute the matrix $L(d) \in \mathbb{F}_2^{R \times N}$, whose rows form the complete set of min-weight logical operators, given $H$, $A$, and $d$.

A naive approach would exhaustively test all ${N \choose d}$ weight-$d$ faults, retaining those with $HF = 0$ and $AF \neq 0$. 
However, this is computationally infeasible; for example, the gross code with perfect measurements would require testing ${144 \choose 12} \approx 10^{17}$ $X$-operators.
The algorithm we propose does not run in polynomial time, but significantly reduces the computation allowing analysis for moderate code sizes.
It leverages two key insights:
\begin{enumerate}
    \item \textbf{Min-weight logical operators form connected components in $G$: }  
    Each minimum-weight logical operator corresponds to a connected component in the graph $G$, where all fault nodes in its support are connected by paths that avoid fault nodes outside the operator's support. 
    This can be proved by contradiction. 
    Assume a logical operator $F$ with weight $d$ has multiple disconnected components, one of which is $F_c$. 
    Decomposing $F$ as $F = F_c \sqcup F_\perp$, we have $HF = 0$, which implies $HF_c = 0$ and $HF_\perp = 0$. 
    This means that $F_c$ is either a stabilizer (making $F_\perp$ a smaller logical operator) or itself a smaller logical operator than $F$, contradicting the assumption.

    \item \textbf{Logical operator enclosure constraint: }  
    A fault $F$ cannot be enclosed by a minimum-weight logical operator if $h_\text{min}(\sigma(F)) + |F| > d$.
    
\end{enumerate}

The first insight enables the use of the decision tree: by selecting any fault node $j$, all size-$d$ connected components enclosing $j$ are nodes at level-$(d-1)$ of the decision tree $\tau(\sigma)$ for $ \sigma = \mathcal{N}(j)$. 
This reduces the search space to $N \cdot (r - 1)^{d-1} \approx 10^9$ for the gross code (where $r$ is the row weight of $H$).
The second insight further reduces the search space by pruning nodes in the decision tree that cannot lead to weight-$d$ logical operators.
We also remove decision tree nodes that correspond to stabilizers.

To manage memory efficiently, the tree exploration is divided into two stages. 
Rather than constructing all layers of the decision tree sequentially, the algorithm generates layers up to a chosen generation $s$ and then identifies logical operators descending from each node at level $s$ individually. 
This is achieved using two algorithms: \alg{descending-logicals}, which takes an initial fault set $F_\text{in}$ with weight $w < d$ and outputs $\mathcal{L}(d)|_{F_\text{in}}$, the set of all weight-$d$ logical operators enclosing $F_\text{in}$; and \alg{find_A}, which constructs the decision tree up to level $s$ and (sequentially) applies \alg{descending-logicals} to each node at that level to determine the descending weight-$d$ logical operators.
(Note that we interchangeably say that a logical operator \emph{descends from}, and that it \emph{encloses the fault set of} a node in the decision tree.)

\begin{algorithm}
\caption{Find $\mathcal{L}(d)|_{F_\text{in}}$, the set of all weight-$d$ logical operators enclosing $F_\text{in}$}
\label{alg:descending-logicals}
\begin{algorithmic}[1]
\Require Error $F_\text{in}$ with weight $w$.
\Require Check matrix $H$, logical action matrix $A$, and minimum weight $d$.
\State Initialize $\mathcal{A}_w \gets \{F_\text{in}\}$ \Comment{Start with the initial fault set.}
\State Initialize empty sets $\mathcal{A}_{r+1} \gets \emptyset$ for $r = w, \dots, d-1$ \Comment{Prepare storage for all generations.}

\For{$r = w$ to $d-1$} \Comment{Iterate through weight layers.}
    \ForAll{$F \in \mathcal{A}_{r}$} \Comment{Explore each fault set in the current layer.}
        \State $\sigma \gets H F$ \Comment{Update the syndrome for $F$.}
        \State Find smallest $i$ such that $\sigma_i = 1$ \Comment{Identify the first unsatisfied check.}
        
        \ForAll{$j \in \mathcal{N}(i)$} \Comment{Iterate over neighboring faults.}
            \If{$F_j = 1$}
                \State \textbf{continue to next $j$} \Comment{Skip if $j$ is already in $F$.}
            \EndIf
            \State $F' \gets F + \{ j \}$ \Comment{Add fault $j$ to the set.}
            \State $\sigma' \gets \sigma + \mathcal{N}(j)$ \Comment{Update the syndrome.}

            \If{$r = d-1$ and $\sigma' = 0$ and $A F' \neq 0$}
                \State Add $F'$ to $\mathcal{A}_{r+1}$ \Comment{Found weight-$d$ logical operator.}
            \ElsIf{$r < d-1$ and $h_\text{min}(\sigma) + r + 1 \leq d$ and $\sigma \neq 0$}
                \State Add $F'$ to $\mathcal{A}_{r+1}$ \Comment{Prune.}
            \EndIf
        \EndFor
    \EndFor
\EndFor
\State \Return $\mathcal{A}_d$ \Comment{Final set $\mathcal{A}_d$ is $\mathcal{L}(d)|_{F_\text{in}}$, all weight-$d$ logical operators enclosing $F_\text{in}$.}
\end{algorithmic}
\end{algorithm}

\begin{algorithm}
\caption{Find $L(d)$ for given $H$, $A$, $d$, and $s$ (tree search separation point)}
\label{alg:find_A}
\begin{algorithmic}[1]
\Require Weight $s$.
\Require Check matrix $H$, logical action matrix $A$, and minimum weight $d$.
\State Initialize $\mathcal{A}_1 \gets \{\{ j \} \mid j = 1, \dots, N\}$ \Comment{Start with single-node fault sets.}
\State Initialize empty sets $\mathcal{A}_{r+1} \gets \emptyset$ for $r = 1, \dots, s-1$ \Comment{Prepare storage for layers.}

\For{$r = 1$ to $s-1$} \Comment{Iterate through layers up to $s$.}
    \ForAll{$F \in \mathcal{A}_{r}$} \Comment{Explore each fault set in the current layer.}
        \State Compute syndrome $\sigma \gets H F$ \Comment{Update the syndrome.}
        \State Find smallest $i$ such that $\sigma_i = 1$ \Comment{Identify the first unsatisfied check.}
        
        \ForAll{$j \in \mathcal{N}(i)$} \Comment{Iterate over neighboring faults.}
            \If{$F_j = 1$}
                \State \textbf{continue to next $j$} \Comment{Skip if $j$ is already in $F$.}
            \EndIf
            \State $F' \gets F + \{ j \}$ \Comment{Add fault $j$ to the set.}
            \State $\sigma' \gets \sigma + \mathcal{N}(j)$ \Comment{Update the syndrome.}
            \If{$\sigma \neq 0$}
                \State Add $F'$ to $\mathcal{A}_{r+1}$ \Comment{Store valid sets for the next layer.}
            \EndIf
        \EndFor
    \EndFor
\EndFor

\State Initialize $\mathcal{A} \gets \emptyset$ \Comment{Prepare to compute logical operators.}
\ForAll{$F \in \mathcal{A}_s$} \Comment{Explore each fault set at separation level $s$.}
    \State Use \alg{descending-logicals} to compute $\mathcal{L}(d)|_{F_\text{in}}$ \Comment{Find descending weight-$d$ logical operators.}
    \State Add elements of $\mathcal{L}(d)|_{F_\text{in}}$ to $\mathcal{A}$
\EndFor

\State Form matrix $L(d)$ by taking each element in $\mathcal{A}$ as a row \Comment{Construct result matrix.}
\State \Return $L(d)$ \Comment{Output matrix of min-weight logical operators.}
\end{algorithmic}
\end{algorithm}

Using this approach with height bound $h_2(\sigma)$ from \eq{tighter-bound}, we found 84 and 1884 distinct distance-$d$ $X$-type logical operators for the [[72,12,6]] and [[144,12,12]] bivariate bicycle codes, respectively.

\subsection{Runtime comparison of height-bound DTD and MaxSAT decoding}
\label{app:maxsat-comparison}

Here, in \fig{hbdtd-vs-maxsat} we compare the runtime required to decode color codes under data qubit noise by height-bound DTD (\sec{min-weight-numerics}) with the MaxSAT decoder from Ref.~\cite{noormandipour2024maxsatdecodersarbitrarycss}, discussed in \sec{maxsat-intro}. 
Both are general qLDPC decoders guaranteed to produce minimum-weight corrections.
The height-bound DTD data (from a single process on an Intel Core i9-10885H CPU, Windows 11, 32GB RAM) is compared with data extracted (by eye) from Fig. 1b of Ref.~\cite{noormandipour2024maxsatdecodersarbitrarycss} (using an award-winning MaxSAT solver Open-WBO~\cite{10.1007/978-3-319-09284-3_33}, run serially on an unspecified CPU).
This crude comparison provides qualitative insights but should be interpreted cautiously, as runtimes were measured on different machines and neither implementation was optimized for absolute runtime. 
Notably, the height-bound DTD appears significantly faster in the low-error-rate regime but similar or even slower than the MaxSAT decoder near the threshold.

\begin{figure}[h]
  \centering
    \includegraphics[width=0.6\textwidth]{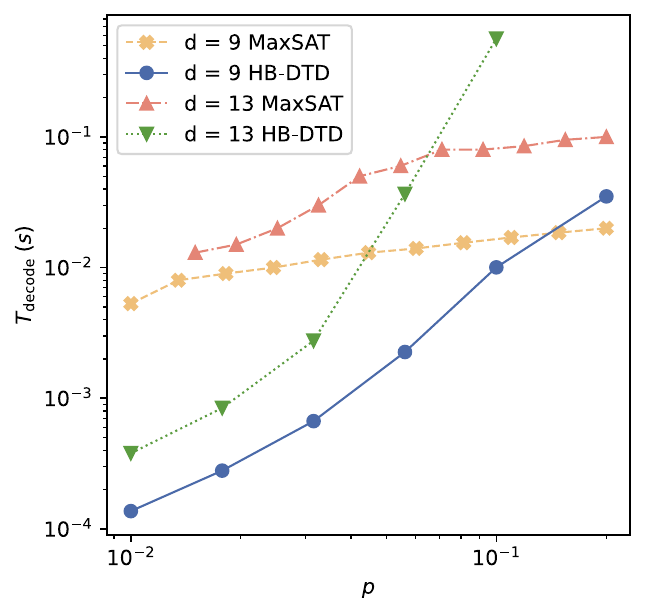}
     \caption{
     Comparison of the provable minimum-weight decoders, height-bound DTD and MaxSAT, for distance-9 and distance-13 color codes. 
     In the low-error-rate regime, height-bound DTD achieves runtimes more than an order of magnitude faster than MaxSAT. 
     At higher error rates, its performance becomes comparable to or worse than MaxSAT. 
     }
  \label{fig:hbdtd-vs-maxsat}
\end{figure}

\end{document}